\documentclass{article}
\usepackage{eurosym}
\usepackage[toc,page]{appendix}
\usepackage{amsfonts}
\usepackage{amsmath}
\usepackage[top=2cm, bottom=2cm, left=2.75cm, right=2.75cm]{geometry}
\usepackage[colorlinks=true, linkcolor=red, citecolor=blue]{hyperref}

\makeatletter

\@addtoreset{equation}{section}
\makeatother
\newtheorem{theorem}{Theorem}

\newtheorem{lemma}[theorem]{Lemma}

\newtheorem{remark}[theorem]{Remark}

\newenvironment{proof}[1][\textit{Proof}]{\noindent\textbf{#1.} }{\ \rule{0.5em}{0.5em}}
\begin{document}

\title{Heat coefficients for magnetic Laplacians on the complex projective
space $\mathbf{P}^{n}(\mathbb{C})$}
\author{K. Ahbli${}^{*}$, A. Hafoud${}^\sharp$, Z. Mouayn${}^{\flat}$}
\maketitle

\begin{abstract}
Denoting by $\Delta_\nu$ the Fubini-Study Laplacian perturbed by a uniform magnetic field strength proportional to $\nu$, this operator has a discrete spectrum consisting on eigenvalues $\beta_m, \ m\in\mathbb{Z}_+$, when acting on bounded functions of the complex projective $n$-space. For the corresponding eigenspaces, we give a new proof for their reproducing kernels by using Zaremba's expansion directly. These kernels are then used to obtain an integral representation for the heat kernel of $\Delta_\nu$. Using a suitable polynomial decomposition of the multiplicity of each $\beta_m$, we write down a trace formula for the heat operator associated with $\Delta_\nu$ in terms of Jacobi's theta functions and their higher order derivatives. Doing so enables us to establish the asymptotics of this trace as $t\searrow 0^+$ by giving the corresponding heat coefficients in terms of Bernoulli numbers and polynomials. The obtained results can be exploited in the analysis of the spectral zeta function associated with $\Delta_\nu$.
\end{abstract}

\section{Introduction}
Let $\Delta$ denote the Laplace-Beltrami operator on a Riemannian manifold $(\mathcal{M},g)$ of dimension $d$, given by 
\begin{equation}
\Delta=-\frac{1}{\sqrt{\det g}}\sum\limits_{i,j}\partial_i\left(g^{ij}\sqrt{\det g}\partial_j\right).
\end{equation}
 The heat kernel is a function $H(x,y,t)\in \mathcal{C}^{\infty}(\mathcal{M}\times \mathcal{M}\times \mathbb{R}_{+})$ that solves the problem 
\begin{equation}
\left( \partial_t +\Delta\right)H=0 \quad \text{ and } \quad \lim\limits_{t\to 0}\int_{\mathcal{M}}H(x,y,t)f(y)dy=f(x)
\end{equation} 
 for any smooth function $f$ of compact support. For $\mathcal{M}$ compact, there exists a complete orthonormal basis $\{\phi_k\}$ in $L^2(\mathcal{M})$, consisting of eigenfunctions of $\Delta$, associated with eigenvalues $0\leq \lambda_0 \leq \lambda_1, \cdots $, such that 
\begin{equation}
H(x,y,t)=\sum\limits_{k=0}^{+\infty}e^{-\lambda_k t}\phi_k(x)\phi_k(y)
\end{equation}
with uniformly convergence for any fixed $t>0$, see \cite{Da1989} for the general theory. The asymptotic expansion $H(x,y,t)$ had be studied by Minakshisudaran and Pleijel \cite{MP1949} and one has \cite{Mi1953}: 
\begin{equation}
H(x,x,t)\sim \frac{1}{(2\sqrt{\pi t})^d} (1+a_1(x)t+a_2(x)t^2+\cdots),\quad t \searrow 0^{+}
\end{equation}
which, after being integrated with respect to the volume form, gives the trace formula
\begin{equation}\label{eq1.3}
\displaystyle\sum_{k=0}^{+\infty} e^{-t\lambda_k} \sim \frac{1}{(2\sqrt{\pi t})^d} (a_0+a_1t+a_2t^2+\cdots),\quad t \searrow  0^{+}.
\end{equation}
The numbers $a_k$ in \eqref{eq1.3} are called \textit{heat coefficients} and may be expressed in terms of geometrical invariants of $\mathcal{M}$. Indeed, $a_0$ is the volume of $\mathcal{M}$, and for $d=2$, $a_1$ is proportional to the Euler characteristic of $\mathcal{M}$. In general, $\{a_k\}$ depend on the curvature tensor $R$ and its successive covariant derivatives meaning that it is not easy to determine them. Indeed, few of these coefficients have been calculated for a general manifold \cite{BGM1971} while in the case of compact symmetric spaces of rank one they can be found in \cite{CW_1976}.

For $\mathcal{M}=\mathbf{P}^{n}(\mathbb{C})$ the complex projective $n$-space which is a real manifold of dimension $2n$ representing the prototype of rank-one complex Riemannian symmetric space of compact-type, the computation of the heat coefficients was discussed in \cite{Byt99} through the spectral zeta functions and via the resolvent of the Laplacian operator in \cite{Po02}. Other methods can be found in \cite{CoVZe} and for further discussion on these coefficients we refer to \cite{BrGi}. 

In \cite{Aw2019} the author has examined the role played by the coefficients $a_k^{(n)}\equiv a_k\left(\mathbf{P}^{n}(\mathbb{C})\right)$ in describing a new class of heat coefficients and then introduced the associated zeta function. He also proved that for the Fubini-Study Laplacian 
\begin{equation}  \label{FSL}
\Delta _{0}:=4(1+\langle z,z\rangle )\sum_{i,j=1}^{n}(\delta _{ij}+z_{i}\bar{%
z}_{j})\frac{\partial ^{2}}{\partial z_{i}\partial \bar{z}_{j}},
\end{equation}
the expansion of the trace of the heat operator
\begin{equation}
\text{Tr}\left(e^{-t\Delta_{0}}\right)\sim \frac{1}{(4\pi t)^n}\sum\limits_{k=0}^{+\infty}a_k^{(n)}t^k, \quad t\searrow 0^+, 
\end{equation}
 can be expressed purely in terms of Jacobi theta functions $\vartheta_2(t)$ and $\vartheta_3(t)$ and their higher order derivatives, depending on the cases $n$ odd and $n$ even respectively. A similar discussion involving these special functions already appeared in \cite{HI2002} where the authors have established an integral representation for the heat kernel 
\begin{equation}\label{IRHK_HI}
H_0(x,y,t)=\frac{e^{n^{2}t}}{2^{n-2}\pi ^{n+1}}\int_{\rho }^{\frac{\pi }{%
2}}\frac{-d(\cos u)}{\sqrt{\cos ^{2}\rho -\cos ^{2}u}}\left( -\frac{1}{\sin u%
}\frac{d}{du}\right) ^{n}\left[ \Theta _{n+1}(t,u)\right]
\end{equation}
associated with $\Delta_0$, where $\rho =d_{FS}(x,y)$ denotes the Fubini-Study distance (given by \eqref{d_FS} below) and
\begin{equation*}
\Theta _{n+1}(t,u):=\sum_{j=0}^{+\infty }e^{-4t(j+\frac{n}{2})^{2}}\cos(2j+n)u.
\end{equation*}
They also  suggested tackling the heat trace asymptotics problem by exploiting \eqref{IRHK_HI}.\\

In the present paper, we deal with similar questions in the context of the complex projective $n$-space and for the following perturbed form of $\Delta_0$:
\begin{equation}
\Delta _{\nu }:=4(1+\langle z,z\rangle )\left( \sum_{i,j=1}^{n}(\delta
_{ij}+z_{i}\bar{z_{j}})\frac{\partial ^{2}}{\partial z_{i}\partial \bar{z_{j}%
}}-\nu \sum_{j=1}^{n}\left( z_{j}\frac{\partial }{\partial z_{j}}-\bar{z_{j}}%
\frac{\partial }{\partial \bar{z_{j}}}\right) -\nu ^{2}\right) +4\nu ^{2}
\label{delta_nu_intro}
\end{equation}
called magnetic Laplacian (with $2\nu\in \mathbb{Z}_+$) as introduced in \cite{HI_2005}, which can also be viewed as the Bochner Laplacian on powers of the Hopf line bundle and of its conjugate \cite{BCC,Pee-Zha}. When acting on the space of bounded functions, $\Delta _{\nu }$ admits a discrete spectrum consisting on eigenvalues called spherical Landau levels ($\nu$ is proportional to the magnetic field strength). For the corresponding eigenspaces $\mathcal{A}_{m}^{\nu}, \ m\in \mathbb{Z}_+$, the authors \cite{DMY} have established formulae for the associated Berezin transforms as functions of $\Delta_0$. 

For these spaces $\mathcal{A}_{m}^{\nu}$, we present a new proof for their reproducing kernels  using Zaremba's expansion directly. These reproducing kernels are then used to establish an integral representation for the heat kernel of $\Delta_\nu$, extending the one in \eqref{IRHK_HI} for $\nu\neq 0$. To obtain the heat coefficients in the asymptotic expansion of the trace of the heat operator $\exp\left(\frac{1}{4}t\Delta_\nu\right)$, we first use a suitable polynomial decomposition of dimensions of eigenspaces $\mathcal{A}_{m}^{\nu}$ to write down a trace formula for this operator. Next, we express this trace in terms of Jacobi's theta functions and their higher order derivatives, and then take the asymptotics of these special functions as $t\searrow 0^+$. Doing so enables us, after straightforward calculations, to find out precise formulae for the heat coefficients in terms of Bernoulli numbers and polynomials.\\

The paper is organized as follows. In Section 2, we recall the geometrical construction of
magnetic Laplacians $\Delta _{\nu }$ on the complex projective $n$-space and we illustrate this construction in the case $n=1$ by explaining the role of Dirac monopoles. In Section 3, we summarize some notations about
spherical harmonics that help us recalling some results on eigenspaces of $\Delta _{\nu }$. For these spaces, we give in Section 4 a new proof for their reproducing kernels. In Section 5, we recall the Cauchy problem for the heat equation associated with $\Delta _{\nu }$. In Section 6 we establish an integral representation for the corresponding heat kernel. In Section 7, we write down an asymptotic expansion of the trace of the heat semigroup $\exp\left(\frac{1}{4}t\Delta _{\nu }\right)$ where the heat coefficients are obtained explicitly. In Section 8, we exhibit these coefficients for specific values of $n\in \{1,2,3,4\}$. In appendix A, we list the definitions of some needed special functions and orthogonal polynomials. 

\section{Magnetic Laplacians $\Delta _{\protect\nu }$}

We here recall from \cite{HI_2005} the construction of magnetic Laplacians $\Delta _{\nu }$ on the complex projective $n$-space, $n\geq 1$. Let $\mathbb{S}^{2n+1}$ be the $(2n+1)-$dimensional unit sphere of $\mathbb{C}^{n+1}.$
Then, the unit circle $U\left( 1\right) \equiv \mathbb{S}^{1}$ acts freely on 
$\mathbb{S}^{2n+1}$ and define the complex projective space by $\mathbf{P}^{n}(\mathbb{C})=\mathbb{S}^{1}\setminus $ $\mathbb{S}^{2n+1} $ to be the set of all complex lines of $\mathbb{C}^{n+1}$. Indeed,
the Hopf fibration $\mathbb{S}^{1}\rightarrow \mathbb{S}^{2n+1}\rightarrow \mathbf{P}^{n}(\mathbb{C})$ defines a principal $U\left(1\right) -$bundle on $\mathbf{P}^{n}(\mathbb{C})$ whose associated complex line is $%
\mathfrak{L}=\left\{ \left( l,z\right) \in \mathbf{P}^{n}\left( \mathbb{C}\right) \times \mathbb{C}^{n+1},z\in l\right\} $. It is endowed with the
Fubini-Study metric $ds_{FS}^{2}$ which reads in a standard charte%
\begin{equation*}
\mathbb{C}^{n}\equiv \left\{ \left( z_{1},...,z_{n+1}\right) \in \mathbb{C}%
^{n+1};z_{n+1}=1\right\}
\end{equation*}
or local coordinates as 
\begin{equation}
ds_{FS}^{2}:=\sum\limits_{i,j=1}^{n}\left( \left( 1+\langle z,z\rangle
\right) \delta _{ij}-z_{i}\overline{z}_{j}\right) dz_{i}\otimes d\overline{z}%
_{j}, \quad \langle z,z\rangle=|z|^2,
\end{equation}%
where $g_{ij}\left( z\right) =\left( 1+\langle z,z\rangle \right)
^{-2}\left( \left( 1+\langle z,z\rangle \right) \delta _{ij}-z_{i}\overline{z%
}_{j}\right) .$ The complex projective $n$-space equipped with this metric
is a Kh\"{a}lerian compacte manifold of complex dimension $n.$ The
associated Laplace-Beltrami operator is given by 
\begin{equation}
\sum\limits_{i,j=1}^{n}g^{ij}\left( z\right) \frac{\partial ^{2}}{\partial
z_{i}\partial \overline{z}_{j}}
\end{equation}%
where $\left( g^{ij}\left( z\right) \right) $ is the inverse of the matrix $%
\left( g_{ij}\left( z\right) \right) .$ In local coordinates $z=\left(
z_{1},...,z_{n}\right) $ $\in $ $\mathbb{C}^{n},$ it reads%
\begin{equation}
\Delta _{0}=4(1+\langle z,z\rangle )\sum_{i,j=1}^{n}(\delta _{ij}+z_{i}\bar{%
z}_{j})\frac{\partial ^{2}}{\partial z_{i}\partial \bar{z}_{j}}.
\label{delta_fs}
\end{equation}%
The Fubini-Study distance is defined by 
\begin{equation}  \label{d_FS}
\cos ^{2}d_{FS}(z,w)=\frac{|1+\langle z,w\rangle |^{2}}{(1+\langle
z,z\rangle )(1+\langle w,w\rangle )}.\qquad
\end{equation}%
Let $\nabla =d+\partial \log \left( 1+\langle z,z\rangle \right) $ be the
unique hermitian connection associated with $ds_{FS}^{2}$ on $\mathfrak{L}.$
Now, for a positive integer $\nu $ let $\mathfrak{L}^{\nu }:=\left( 
\overline{\mathfrak{L}^{\ast }}\right) ^{\otimes \nu }\otimes \left( 
\mathfrak{L}^{\ast }\right) ^{\otimes \nu }$ the complex line bundle on $\mathbf{P}%
^{n}(\mathbb{C}).$ Then, the corresponding hermitian connection reads 
\begin{equation}
\nabla _{\nu }=d+\nu \left( \partial -\overline{\partial }\right) \partial
\log \left( 1+\langle z,z\rangle \right) .
\end{equation}%
Next, consider the operator 
\begin{equation}
\Delta _{\nu }:=-\left( \nabla _{\nu }\right) ^{\ast }\nabla _{\nu }
\label{delta_nu0}
\end{equation}%
acting on the space of smooth sections $\Gamma _{n,\nu }^{\infty }:=$ $%
C^{\infty }\left( \mathbf{P}^{n}(\mathbb{C}) ,\mathfrak{L}^{\nu }\right) ,$
which is also known as the Bochner Laplacian on hermitian line bundles
parametrized by the magnetic field strength $\nu .$ Precisely, in the local
coordinates the operator $\Delta _{\nu }$ takes the form 
\begin{equation}
\Delta _{\nu }=4(1+\langle z,z\rangle )\left( \sum_{i,j=1}^{n}(\delta
_{ij}+z_{i}\bar{z_{j}})\frac{\partial ^{2}}{\partial z_{i}\partial \bar{z_{j}%
}}-\nu \sum_{j=1}^{n}\left( z_{j}\frac{\partial }{\partial z_{j}}-\bar{z_{j}}%
\frac{\partial }{\partial \bar{z_{j}}}\right) -\nu ^{2}\right) +4\nu ^{2}
\label{delta_nu1}
\end{equation}%
which, according to \cite{HI_2005}, will be called a magnetic Laplacians on $\mathbf{P}^{n}(\mathbb{C}) $. The dependence of this operator on $n$ is omitted.

\subsection{Example}

\label{example1} On elements of $\Gamma _{1,\nu }^{\infty }$, which are
smooth sections of the $U\left( 1\right) $-bundle with the first Chern class 
$\nu $ acts the Hamiltonian $H_{\nu }$ of the Dirac (point) monopole in $%
\mathbb{R}^{3}$ with magnetic charge $\nu $ (in suitable units). Indeed,
eigenfunctions of this monopole have been identified as \textit{sections} by
Wu and Yang \cite{WY} and are known as \textit{monopole harmonics.} Their
explicit expression in the coordinate $z$ are given below by \eqref{phi_k}.
For more information on Dirac monopoles see \cite{Sh}. The restriction $2\nu
\in \left\{ 1,2,3,...\right\} $ results from Dirac's quantization condition
for monopole charges, which requires that the total flux of the magnetic
field across a closed surface must be an integer mulitple of a universal
constant. It can also be understood in the context of cohomology groups for
hermitian line bundles \cite{Hir} or as the Weil-Souriau-Kostant
quantization condition \cite{SW}. In the stereographic coordinate $z\in 
\mathbb{C\cup }\left\{ \infty \right\} \equiv \mathbb{S}^{2}\equiv \mathbf{P}^{1}(\mathbb{C}) $ (and suitable units) this Hamiltonian reads (\cite{FV}, \text{p.598})
: 
\begin{equation}
L_{\nu }=-(1+\left\vert z\right\vert ^{2})\left( (1+\left\vert z\right\vert
^{2})\frac{\partial ^{2}}{\partial z\partial \bar{z}}+\nu \left( z_{j}\frac{%
\partial }{\partial z_{j}}-\bar{z_{j}}\frac{\partial }{\partial \bar{z_{j}}}%
\right) -\nu ^{2}\right) -\nu ^{2}\equiv -\frac{1}{4}\Delta _{\nu }.
\label{H_nu}
\end{equation}%
The stereorgraphic projection bridges between the monopole system and the Landau
system which describes spinless charged particles in perpendicular
homogeneous magnetic fields (\cite{Dun}, p.240). Precisely, to determine
eigenstates of the monopole Hamiltonian \eqref{H_nu} one proceeds exactly as
for the Landau Hamiltonian of the Euclidean setting. This leads, for each
fixed integer $m\in \mathbb{Z}_{+}$, to a finite dimensional $L^{2}$
eigenspace whose orthonormal basis vectors $\left\{ \Phi _{k}^{\nu
,m}\right\} ,-m\leq k\leq 2\nu +m,$ are given by \cite{Mo} :

\begin{equation}
\Phi _{k}^{\nu ,m}(z):=\sqrt{\frac{\left( 2\nu +2m+1\right) \left( 2\nu
+m\right) !m!}{(m+k)!(2\nu +m-k)!}}\left( 1+z\overline{z}\right) ^{-\nu
}z^{k}P_{m}^{\left( k,2\nu -k\right) }\left( \frac{1-z\overline{z}}{1+z%
\overline{z}}\right) ,  \label{phi_k}
\end{equation}%
$P_{m}^{\left( \alpha ,\beta \right) }\left( \cdot \right) $ being the
Jacobi polynomial \cite{AAR}. These eigenfunctions \eqref{phi_k} are associated
with the eigenvalue
\begin{equation}  \label{tau_m}
\tau _{m}:=\left( 2m+1\right) \nu +m\left( m+1\right)
\end{equation}
called a spherical Landau level.

\section{Spaces of bounded eigenfunctions of $\Delta _{\protect\nu }$}
In order to summarize some needed results \cite{HI_2005} about eigenspaces of $\Delta_\nu$, we first need to fix some notations. 
\subsection{Spherical harmonics}
Let $\mathcal{P}(\mathbb{C}^{n})$ denote the space of polynomials in the
independent variable $z$ and $\bar{z}$ of $\mathbb{C}^{n}$. Elements of this
space can be written in the form 
\begin{equation}
u(z,\bar{z})=\sum_{|\alpha |\leq k}\sum_{|\beta |\leq l}c_{\alpha ,\beta
}z^{\alpha }\bar{z}^{\beta },\qquad c_{\alpha ,\beta }\in \mathbb{C},\quad
\alpha ,\beta \in \mathbb{Z}_{+}^{n},
\end{equation}%
for non-negative integers $k$ and $l$, where standard multi-index is used.

The subspace of $\mathcal{P}(\mathbb{C}^{n})$ composed of polynomials that
are homogeneous of degree $p$ in $z$ and of degree $q$ in $\bar{z}$ will be
denoted by $\mathcal{P}_{p,q}(\mathbb{C}^{n})$. The dimension of $\mathcal{P}%
_{p,q}(\mathbb{C}^{n})$ is given by 
\begin{equation}
\delta (n,p,q)=%
\begin{pmatrix}
p+n-1 \\ 
p-1%
\end{pmatrix}%
\begin{pmatrix}
q+n-1 \\ 
q-1%
\end{pmatrix}%
\end{equation}%
The subspace of $\mathcal{P}_{p,q}(\mathbb{C}^{n})$ composed of harmonic
elements, that is, elements in the kernel of the complex Laplacian 
\begin{equation}
\sum_{j=1}^{n}\frac{\partial ^{2}}{\partial
z_{j}\partial \bar{z}_{j}}
\end{equation}%
will be denoted by $\mathfrak{H}_{p,q}(\mathbb{C}^{n})$. 

The set of restrictions of elements of $\mathfrak{H}_{p,q}(\mathbb{C}^{n})$ to the unit
sphere $\mathbb{S}^{2n-1}=\{\zeta \in \mathbb{C}^{n},\langle \zeta ,\zeta
\rangle =1\},$ denoted by $\mathcal{H}\left( p,q\right) $, is called the
space of complex spherical harmonics of degree $p$ in $z$ and degree $q$ in $%
\bar{z}$. Note that $\mathcal{H}\left( p,0\right) $ consists of holomorphic
polynomials, and $\mathcal{H}\left( 0,q\right) $ consists of polynomials
whose complex conjugates are holomorphic. 

The dimensions of spaces $\mathcal{%
H}\left( p,q\right) $, denoted by $d(n,p,q)$, are given by 
\begin{equation}
d(n,p,q)=\delta (n,p,q)-\delta (n,p-1,q-1),\quad p,q\neq 0,
\end{equation}%
\begin{equation}
d(n,p,0)=\delta (n,p,0)\qquad \mathrm{and}\qquad d(n,0,q)=\delta (n,0,q).
\end{equation}%
For $n=1$, $d(n,p,0)=d(n,0,q)=1$, but $\mathcal{H}\left( p,q\right) =\{0\}$
if both $p>0$ and $q>0$. It is a standard fact that $\mathcal{H}%
\left( p,q\right) $ are pairwise orthogonal in $L^{2}(\mathbb{S}%
^{2n-1},d\omega )$ where $d\omega $ is the uniform measure on the sphere.
\subsection{Bounded eingenfunctions of $\Delta _{\nu }$}
For $\lambda \in \mathbb{C}$,
we set $\Lambda _{n,\nu }\left( \lambda \right) :=n^{2}-\lambda ^{2}+4\nu^{2}$ and consider the equation 
\begin{equation}
\Delta _{\nu }F(z)=\Lambda _{n,\nu }\left( \lambda \right) F(z),
\end{equation}%
where $F$ is a bounded function on $\mathbb{C}^{n}$. Define the eigenspace 
\begin{equation}
\mathcal{A}_{m}^{\nu }:=\{F\in L^{\infty }(\mathbb{C}^{n}),\text{ }\Delta
_{\nu }F=\Lambda _{n,\nu }\left( \lambda \right) F\}\subset L^{2}(\mathbb{C}%
^{n},d\mu _{n}),  \label{A_m,nu}
\end{equation}%
where 
\begin{equation}
d\mu _{n}(z):=(1+\langle z,z\rangle )^{-(n+1)}d\mu \left( z\right) ,
\end{equation}%
$d\mu \left( z\right) $ being the Lebesgue measure on $\mathbb{C}^{n}.$ The
eigenspace $\mathcal{A}_{m}^{\nu }=\{0\}$ if 
\begin{equation}
\lambda \not\in \{\xi \in \mathbb{C},\ \frac{1}{2}(n\pm \xi )+\nu \in 
\mathbb{Z}_{-}\}\cup \{\xi \in \mathbb{C},\ \frac{1}{2}(n\pm \xi )-\nu \in 
\mathbb{Z}_{-}\}.
\end{equation}%
otherwise it is not trivial if and only if $\lambda $ has the form $\lambda
=\pm (2(m+\nu )+n)$ for some $m\in \mathbb{Z}_{+}$. Note that when $n=1$ and 
$\lambda =\pm (2(m+\nu )+1)$ the above parametrization of the eigenvalue of $%
\Delta _{\nu }$ gives that $\frac{-1}{4}\Lambda _{1,\nu }\left( \lambda
\right) =\tau _{m}$ (where $\tau _{m}$ was given by \eqref{tau_m}) as
expected form the example in Subsection \ref{example1}. For $n\geq 1$ and
under the condition $\lambda =\pm (2(m+\nu )+n)$ one gets the eigenvalue 
\begin{equation}
\beta _{m}:=\Lambda _{n,\nu }\left( \pm (2(m+\nu )+n)\right) =-4(m+\nu
)(m+\nu +n)+4\nu ^{2}
\end{equation}%
and any function $F$ in $\mathcal{A}_{m}^{\nu }$ admits the expansion 
\begin{equation}
F(r\omega )=\frac{1}{\left( 1+r^{2}\right) ^{(m+\nu )}}\sum_{\substack{ %
0\leq p\leq m  \\ 0\leq q\leq m+2\nu }}r^{p+q}{}_{2}F_{1}\left( 
\begin{array}{c}
p-m,q-m-2\nu \\ 
n+p+q%
\end{array}%
\mid -r^{2}\right) \sum_{j=1}^{d(n,p,q)}a_{j}^{\nu ,p,q}\text{ }%
h_{p,q}^{j}(\omega ,\bar{\omega}),
\end{equation}%
where $r>0,\ \omega \in \mathbb{S}^{2n-1},\ a_{j}^{\nu ,p,q}$ are constant
complex numbers, $_{2}F_{1}$ is the Gauss hypergeometric function (see Appendix \ref{appendix1}) and $\left\{ h_{p,q}^{j}\right\}
_{j=0}^{d(n,p,q)}$ is an orthonormal basis of $\mathcal{H}\left( p,q\right) $%
. Note that $F$ satisfies the growth condition 
\begin{equation}
\lim_{r\rightarrow \infty }F(r\omega )=\sum_{0\leq p\leq m}\frac{\Gamma
(m-p+1)\Gamma (n+2p+2\nu )}{(-1)^{p-m}\Gamma (m+n+p+2\nu )}%
\sum_{j=1}^{d(n,p,p+2\nu )}a_{j}^{\nu ,p}\text{ }h_{p,p+2\nu }^{j}(\omega ,%
\bar{\omega})
\end{equation}%
where we wrote $a_{j}^{\nu ,p}=$ $a_{j}^{\nu ,p,p+2\nu }$. In particular, $\mathcal{A}_m^\nu$ is finite-dimensional and has the following orthonormal basis \cite[p.150]{HI_2005}: 
\begin{equation}\label{basis}
\Phi _{j}^{p,q,m}(z)=\gamma _{p,q}^{n,\nu ,m}(1+\left\langle
z,z\right\rangle )^{-m-\nu }{}_{2}F{}_{1}\left( 
\begin{array}{c}
p-m,q-m-2\nu \\ 
n+p+q%
\end{array}%
\big|-\left\langle z,z\right\rangle \right) h_{p,q}^{j}(z,\bar{z})
\end{equation}%
where 
\begin{equation*}
\gamma _{p,q}^{n,\nu ,m}=\sqrt{\frac{2(2m+2\nu +n)\Gamma (m+q+n)\Gamma
(m+2\nu +p+n)}{\Gamma ^{2}(n+p+q)\Gamma (m-p+1)\Gamma (m+2\nu -q+1)}}.
\end{equation*}%
for varying $0\leq p\leq m,$ $0\leq q\leq m+2\nu $ and $j=1,...,d(n;p,q)$.
Here $\gamma _{p,q}^{n,\nu ,m}=\left\Vert \Phi _{j}^{p,q,m}(z)\right\Vert
^{-1}$ where the norm is taken in $L^{2}\left( \mathbb{C}^{n},d\mu
_{n}\left( z\right) \right) $. The reproducing kernel of $\mathcal{A}_m^\nu$ is given by \cite[p.148]{HI_2005}:
\begin{equation}
K_{\nu ,m}(z,w):=\frac{(2m+2\nu +n)\Gamma (m+n+2\nu )}{\pi ^{n}\Gamma
(m+2\nu +1)}\left( \frac{\left( 1+\langle w,z\rangle \right) ^{2}}{%
(1+|z|^{2})(1+|w|^{2})}\right) ^{\nu }P_{m}^{(n-1,2\nu )}(\cos 2d_{FS}(z,w)),
\label{rep_kern}
\end{equation}%
where $P_{m}^{(n-1,2\nu )}(\cdot )$ is the Jacobi polynomial of parameters $n-1 $ and $2\nu $ (see \eqref{jacobi}).
\section{A new proof for reproducing kernels of $\mathcal{A}_{m}^{\protect\nu }$}
In order to obtain the reproducing kernel \eqref{rep_kern}, the authors \cite{HI_2005} have used the Zaremba expansion (\cite{Szafraniec}) of the reproducing kernel to compute $K_{\nu ,m}(z,0)$ as well as to
establish a two-point invariance-type property of this kernel, involving the transitive action of the group $SU\left(n+1\right) $ on $\mathbf{P}^{n}\left( \mathbb{C}\right) .$ Here, we shall retrieve the same result by using the Zaremba expansion directly. 
\begin{lemma}
The functions \eqref{basis} satisfy 
\begin{equation}  \label{RK}
\begin{split}
\sum_{p=0}^{m}\sum_{q=0}^{m+2\nu }\sum_{j=1}^{d(n,p,q)}\Phi _{m,p,q}^{j}(z)%
\overline{\Phi _{m,p,q}^{j}(w)}&=\frac{(2m+2\nu +n)\Gamma (m+2\nu +n)}{\pi
^{n}\Gamma (m+2\nu +1)}\left( \frac{\left( 1+\langle w,z\rangle \right) ^{2}%
}{(1+|z|^{2})(1+|w|^{2})}\right) ^{\nu } \\
&\times P_{m}^{(n-1,2\nu )}\left( \frac{2|1+\langle z,w\rangle |^{2}}{%
(1+|z|^{2})(1+|w|^{2})}-1\right) ,
\end{split}
\end{equation}
for every $z,w\in \mathbb{C}^{n}$.
\end{lemma}

\begin{proof}
Denoting the sum in the l.h.s of \eqref{RK} by $K_{\nu ,m}(z,w)$
and replacing the orthonormal basis \eqref{basis} of $\mathcal{A}_m^\nu$ by their expressions, we obtain 
\begin{equation}  \label{repr_kern_1}
\begin{split}
K_{\nu ,m}(z,w)&=\left( (1+\left\langle z,z\right\rangle )\left( 1+\left\langle
w,w\right\rangle \right) \right) ^{-\nu -m}\sum_{\substack{ 1\leq j\leq
d(n,p,q)  \\ 0\leq p\leq m,\ 0\leq q\leq m+2\nu }}\left( \gamma
_{p,q}^{n,\nu ,m}\right) ^{2}h_{p,q}^{j}(z,\bar{z})\overline{h_{p,q}^{j}(w,%
\bar{w})} \\
&\times _{2}F{}_{1}\left( 
\begin{array}{c}
p-m,q-m-2\nu \\ 
n+p+q%
\end{array}%
\big|-\left\langle z,z\right\rangle \right) {}_{2}F{}_{1}\left( 
\begin{array}{c}
p-m,q-m-2\nu \\ 
n+p+q%
\end{array}%
\big|-\left\langle w,w\right\rangle \right)
\end{split}%
\end{equation}%
which can also be written as 
\begin{equation}  \label{Kern_2F1}
\begin{split}
K_{\nu ,m}(z,w)& =\left( (1+\left\langle z,z\right\rangle )\left(
1+\left\langle w,w\right\rangle \right) \right) ^{-\nu
-m}\sum_{p=0}^{m}\sum_{q=0}^{m+2\nu }\left( \gamma _{p,q}^{n,\nu ,m}\right)
^{2}\mathfrak{S}_{n;p,q}^{z,w} \\
& \times {}_{2}F{}_{1}\left( 
\begin{array}{c}
p-m,q-m-2\nu \\ 
n+p+q%
\end{array}%
\big|-\left\langle z,z\right\rangle \right) {}_{2}F{}_{1}\left( 
\begin{array}{c}
p-m,q-m-2\nu \\ 
n+p+q%
\end{array}%
\big|-\left\langle w,w\right\rangle \right)
\end{split}%
\end{equation}%
where 
\begin{equation}  \label{sum_hs}
\begin{split}
\mathfrak{S}_{n;p,q}^{z,w}&:=\sum_{\substack{ 1\leq j\leq d(n,p,q)}}%
h_{p,q}^{j}(z,\bar{z})\overline{h_{p,q}^{j}(w,\bar{w})} \\
&=\left( |z||w|\right) ^{p+q}\sum_{\substack{ 1\leq j\leq d(n,p,q)}}%
h_{p,q}^{j}\left( \frac{z}{|z|},\frac{\bar{z}}{|z|}\right) \overline{%
h_{p,q}^{j}\left( \frac{w}{|w|},\frac{\bar{w}}{|w|}\right) }.
\end{split}%
\end{equation}%
As usual, in the spherical harmonics setting, this last sum can be computed
via the Koornwinder's formula \cite{Koor_1972} as
\begin{equation*}
\mathfrak{S}_{n;p,q}^{z,w}=\frac{\Gamma (n)d(n,p,q)}{2\pi ^{n}}\left(
|z||w|\right) ^{p+q}\left\vert \left\langle \frac{z}{|z|},\frac{w}{|w|}%
\right\rangle \right\vert ^{\left\vert p-q\right\vert }e^{i\left( p-q\right)
\arg \left( \left\langle \frac{z}{|z|},\frac{w}{|w|}\right\rangle \right) }%
\frac{P_{\min \left( p,q\right) }^{(n-2,\left\vert p-q\right\vert )}\left(
2\left\vert \left\langle \frac{z}{|z|},\frac{w}{|w|}\right\rangle
\right\vert ^{2}-1\right) }{P_{\min \left( p,q\right) }^{(n-2,\left\vert
p-q\right\vert )}(1)}.
\end{equation*}
Now, using the notation $R_{p,q}^{\gamma }(\xi )$ in \eqref{Zernike_poly}, $\mathfrak{S}_{n;p,q}^{z,w}$ also reads 
\begin{equation}  \label{Koor_sum}
\mathfrak{S}_{n;p,q}^{z,w}=(2\pi ^{n})^{-1}\Gamma
(n)d(n,p,q)(|z||w|)^{p+q}R_{p,q}^{n-2}\left( \left\langle \frac{z}{|z|},%
\frac{w}{|w|}\right\rangle \right) .
\end{equation}%
On the other hand, we also need to write the terminating Gauss hypergeometric
functions $_{2}F{}_{1}$ in \eqref{Kern_2F1} in terms of Zernike
polynomials \eqref{Zernike_poly}. For that, we use \eqref{2F1-R} to obtain 
\begin{equation}  \label{2F1_Zern}
{}_{2}F{}_{1}\left( 
\begin{array}{c}
p-m,q-m-2\nu \\ 
n+p+q%
\end{array}%
\big|-|z|^{2}\right) =\left( 1+|z|^{2}\right) ^{m+\nu -\frac{p+q}{2}%
}R_{m-p,m-q+2\nu }^{n+p+q-1}\left( \left( 1+|z|^{2}\right) ^{-\frac{1}{2}%
}\right) .
\end{equation}%
Now, by replacing \eqref{Koor_sum} and \eqref{2F1_Zern} in \eqref{Kern_2F1},
the sum \eqref{repr_kern_1} takes the form 
\begin{equation}  \label{Kern2}
\begin{split}
K_{\nu ,m}(z,w)& =\sum_{p=0}^{m}\sum_{q=0}^{m+2\nu }\left[ (2\pi
^{n})^{-1}\Gamma (n)\left( \gamma _{p,q}^{n,\nu ,m}\right) ^{2}\right]
d(n,p,q)\left( \frac{\left( |z||w|\right) ^{2}}{\left( 1+|z|^{2}\right)
\left( 1+|w|^{2}\right) }\right) ^{\frac{1}{2}\left( p+q\right) } \\
& \times R_{m-p,m-q+2\nu }^{n+p+q-1}\left( \left( 1+|z|^{2}\right) ^{-\frac{1%
}{2}}\right) R_{m-p,m-q+2\nu }^{n+p+q-1}\left( \left( 1+|w|^{2}\right) ^{-%
\frac{1}{2}}\right) R_{p,q}^{n-2}\left( \left\langle \frac{z}{|z|},\frac{w}{%
|w|}\right\rangle \right)
\end{split}%
\end{equation}%
with
\begin{equation}
(2\pi ^{n})^{-1}\Gamma (n)\left( \gamma _{p,q}^{n,\nu ,m}\right) ^{2}=\beta
_{\nu ,n,m}\frac{(-1)^{p+q}(-m)_{p}(-m-2\nu )_{q}(n+m+2\nu )_{p}(n+m)_{q}}{%
(n)_{p+q}(n)_{p+q}},
\end{equation}%
where the constant
\begin{equation}
\beta _{\nu ,n,m}=\frac{(2m+2\nu +n)\Gamma (n+m)\Gamma (n+m+2\nu )}{\pi
^{n}\Gamma (n)m!(m+2\nu )!}
\end{equation}%
was obtained by using identities \eqref{binom} and \eqref{PochhammerGG}. Therefore, \eqref{Kern2} becomes
\begin{equation*}
\begin{split}
K_{\nu ,m}(z,w)& =\beta _{\nu ,n,m}\sum_{\substack{ 0\leq p\leq m  \\ 0\leq
q\leq m+2\nu }}\frac{(-1)^{p+q}(-m)_{p}(-m-2\nu )_{q}(n+m+2\nu )_{p}(n+m)_{q}%
}{(d(n,p,q))^{-1}(n)_{p+q}(n)_{p+q}}\left( \frac{|z|^{2}|w|^{2}}{%
(1+|z|^{2})(1+|w|^{2})}\right) ^{\frac{p+q}{2}} \\
& \times R_{m-p,m-q+2\nu }^{n+p+q-1}\left( \left( 1+|z|^{2}\right) ^{-\frac{1%
}{2}}\right) R_{m-p,m-q+2\nu }^{n+p+q-1}\left( \left( 1+|w|^{2}\right) ^{-%
\frac{1}{2}}\right) R_{p,q}^{n-2}\left( \left\langle \frac{z}{|z|},\frac{w}{%
|w|}\right\rangle \right) .
\end{split}%
\end{equation*}%
We now are in position to apply the addition formula for Zernike
polynomials due to \v{S}apiro \cite{Sapiro1968} and Koornwinder \cite%
{Koorn1972}: 
\begin{equation*}
\begin{split}
R_{k,l}^{\gamma }\left( z_{1}\bar{z}_{2}+(1-z_{1}\bar{z}_{1})^{\frac{1}{2}%
}(1-z_{2}\bar{z}_{2})^{\frac{1}{2}}y\right) &
=\sum_{p=0}^{k}\sum_{q=0}^{l}\tau_{p,q,\gamma -1}\frac{%
(-1)^{p+q}(-k)_{p}(-l)_{q}(l+\gamma +1)_{p}(k+\gamma +1)_{q}}{(\gamma
+1)_{p+q}(\gamma +1)_{p+q}} \\
& \times \left[ (1-z_{1}\bar{z}_{1})(1-z_{2}\bar{z}_{2})\right] ^{\frac{p+q}{%
2}}R_{k-p,l-q}^{\gamma +p+q}(z_{1})\overline{R_{k-p,l-q}^{\gamma +p+q}(z_{2})%
}R_{p,q}^{\gamma -1}(y),
\end{split}%
\end{equation*}%
where 
\begin{equation}
\tau_{p,q,\gamma}:=\frac{(p+q+\gamma +1)(\gamma +1)_{p}(\gamma +1)_{q}}{%
(\gamma +1)p!q!}.
\end{equation}%
Indeed, we specialize this formula for $z_{1}=\left( 1+|z|^{2}\right)^{-%
\frac{1}{2}}, \, z_2=\left( 1+|w|^{2}\right)^{-\frac{1}{2}},\ y=\left\langle 
\frac{z}{|z|},\frac{w}{|w|}\right\rangle ,\ k=m, \ l=m+2\nu $, and $\gamma
=n-1$ to obtain, after replacement, the expression
\begin{equation}\label{RK-R}
K_{\nu ,m}(z,w)=\beta _{\nu ,n,m}R_{m,m+2\nu }^{n-1}\left( \frac{1+\langle
z,w\rangle }{(1+|z|^{2})^{\frac{1}{2}}(1+|w|^{2})^{\frac{1}{2}}}\right) .
\end{equation}%
Using again \eqref{Zernike_poly}, equation \eqref{RK-R} can be rewritten as: 
\begin{equation}  \label{last_kern}
K_{\nu ,m}(z,w)=\frac{(2m+2\nu +n)\Gamma (m+2\nu +n)}{\pi ^{n}\Gamma (m+2\nu
+1)}\left( \frac{\left( 1+\langle w,z\rangle \right) ^{2}}{%
(1+|z|^{2})(1+|w|^{2})}\right) ^{\nu }P_{m}^{(n-1,2\nu )}\left( \frac{%
2|1+\langle z,w\rangle |^{2}}{(1+|z|^{2})(1+|w|^{2})}-1\right).
\end{equation}%
Recalling \eqref{d_FS}, we may also write \eqref{last_kern}
in the form \eqref{RK}. This ends the proof.
\end{proof}
\begin{remark}
For $\nu =0$, the Fubini-Study Laplacian $\Delta _{0}\equiv\Delta _{FS}$ has a discrete spectral decomposition
with eigenvalues $\beta_m =(-4m(m+n))_{m\geq 0}$. Besides, each eigenspace is finite-dimensional and has a
orthonormal basis given by homogeneous spherical harmonics of degree zero. Setting $\nu =0$ in \eqref{rep_kern}, then the kernel of the orthogonal projection from $L^{2}(\mathbb{C}^{n},d\mu _{n})$ onto the $m$-th eigenspace of $\Delta_{0}$ reduces to
\begin{equation}
K_{0,m}(z,w)=\pi ^{-n}\left( 2m+n\right) \frac{(m+n-1)!}{m!}%
P_{m}^{(n-1,0)}(\cos 2d_{FS}(z,w)),  \label{psi_n}
\end{equation}%
in agreement with the result of Koornwinder \cite[Theorem 3.8]{Koo}.
\end{remark}
\section{The heat equation associated with $\Delta _{\protect\nu }$}

We here consider the Cauchy problem 
\begin{equation}  \label{Cauchy_problem}
\partial_t\varphi (t,z)=\Delta _{\nu } \varphi(t,z), \quad
t>0, \ z\in\mathbb{C}^n,
\end{equation}
with the initial condition $\varphi(0,z)=g(z)$, $g\in \mathcal{E}^{\nu,\infty}:=\left(\bigoplus_{m=0}^{+\infty}\mathcal{A}_{m}^{\nu}\right)\cap L^{\infty}(\mathbb{C}^n)$. The solution is given by (\cite{H_2002}):
\begin{equation}\label{sol_CP}
\varphi(t,z)=\int_{\mathbb{C}^n}H_{\nu}(t,z,w)g(w)d\mu_n(w)
\end{equation}
where 
\begin{equation}  \label{heat_kernel_series}
H_{\nu }(t,z,w)=N_t^{(n,\nu)}(z,w)\sum\limits_{m=0}^{+\infty }(2m+2\nu
+n)\frac{\Gamma (m+2\nu +n)}{\Gamma (m+2\nu +1)}e^{-4t(m+\nu +\frac{n}{2}%
)^2}P_{m}^{(n-1,2\nu )}\left( \cos 2d_{FS}(z,w)\right)
\end{equation}
is the heat kernel associated with $\Delta_\nu$ with the prefactor 
\begin{equation}
N_t^{(n,\nu)}(z,w):=\pi ^{-n}\left(\frac{(1+\langle w,z\rangle )^2}{(1+|z|^{2})(1+|w|^{2})} \right)^{\nu }e^{4t(\nu^2+\frac{n^2}{4})} .
\end{equation}
Precisely, the following facts have been proved \cite{H_2002}:
\begin{enumerate}
\item[(i)] The series \eqref{heat_kernel_series} defining $H_{\nu }(t,z,w)$ is normally convergent. Further, for every $k\in%
\mathbb{Z}_+$ and every multi-indices $\alpha,\beta\in\mathbb{Z}_+^n$ we
have that 
\begin{equation}
\sum\limits_{m=0}^{+\infty}\sup\limits_{(z,w)\in\mathbb{C}^n\times\mathbb{C}%
^n}\left|\partial_t^k D_{z,w}^{\alpha ,\beta} \left[e^{\beta_m t}K_{\nu,m}(z,w)\right] \right|\leq
C(t) < +\infty
\end{equation}

\item[(ii)] The function $]0,+\infty [\times \mathbb{C}^n\times\mathbb{C}^n
\rightarrow \mathbb{C}, \ (t,z,w)\mapsto H_{\nu} (t,z,w)$ is $C^{\infty}$
and satisfies the equation \eqref{Cauchy_problem}.

\item[(iii)] To see that the function $\varphi(t,z)$ in \eqref{sol_CP} is a solution of the equation \eqref{Cauchy_problem}, one may write it as 
\begin{equation}
\varphi(t,z)=\sum\limits_{m=0}^{+\infty}e^{\beta_m t} \int_{\mathbb{C}%
^n}K_{\nu,m}(z,w)g(w)d\mu_n(w).
\end{equation}
\end{enumerate}
Recalling that $g\in \mathcal{E}^{\nu,\infty}$, it decomposes as 
\begin{equation}  \label{g_series_gm}
g(z)=\sum\limits_{m=0}^{+\infty}g_m(z), \quad g_m\in \mathcal{A}_{m}^{\nu}.
\end{equation}
Moreover, the component functions $\{g_m\}$ satisfy
\begin{equation}
\int_{\mathbb{C}^n}K_{\nu,m}(z,w)g(w)d\mu_n(w)=g_m(z).
\end{equation}
Since the spaces $\mathcal{A}_{m}^{\nu}$ are pairwise orthogonal in $L^2(%
\mathbb{C}^n,d\mu_n(z))$, it follows that 
\begin{equation}  \label{varphi_series_gm}
\varphi(t,z)=\sum\limits_{m=0}^{+\infty}e^{\beta_m t} g_m(z)
\end{equation}
from which one can prove that 
\begin{equation}
\lim\limits_{t\to 0} \Vert \varphi(t,\cdot)-g(\cdot)\Vert_{L^2(\mathbb{C}%
^n,d\mu_n(z))}=0.
\end{equation}
Finally, from \eqref{varphi_series_gm} and \eqref{g_series_gm} it is clear
that $\varphi(t,z)$ obey the initial condition $\varphi(0,z)=g(z)$.\newline
\begin{remark}
For $\nu =0$, the spectral theorem implies that for any suitable function $U$, the operator $U(\Delta_{0})$ is an integral operator whose kernel is given by
\begin{equation}
\sum\limits_{m=0}^{+\infty }U(-4m(m+n))K_{0,m}(z,w).  \label{sum_w_psi}
\end{equation}%
In particular, one can write the kernel of the heat semigroup $\exp\left(t\Delta _{0}\right)$ through the series expansion%
\begin{equation}
H_{0}\left( t;z,w\right) =\pi ^{-n}\sum_{m\geq 0}\left( 2m+n\right) \frac{%
(m+n-1)!}{m!}e^{-4m(m+n)t}P_{m}^{(n-1,0)}(\cos 2d_{FS}(z,w)),
\label{heat_kernel_nu=0}
\end{equation}
see \cite[p.873]{HI2002}.
\end{remark}

\section{An integral representation for $H_{\nu }(t,z,w)$}
In this section, we extend the formula \eqref{IRHK_HI} with respect to the parameter $\nu$ as follows.

\begin{theorem}
Let $n\geq 1$, $2\nu \in\mathbb{Z}_+$ and $\rho =d_{FS}(z,w)$. Then, 
\begin{equation}\label{IRHK-AMH}
\begin{split}
H_{\nu }(t,z,w)&=\frac{e^{4t(\nu ^{2}+\frac{n^{2}}{4})}}{2^{n+2\nu -1}\pi
^{n}\Gamma \left( 2\nu +\frac{1}{2}\right) }\left( \frac{(1+\langle
z,z\rangle )(1+\langle w,w\rangle )}{(1+\langle z,w\rangle )^{2}}\right)
^{\nu }\\
&\times \int_{\rho }^{\frac{\pi }{2}}\frac{-d(\cos u)}{\left( \cos ^{2}\rho
-\cos ^{2}u\right) ^{\frac{1}{2}-2\nu }}\left( -\frac{1}{\sin u}\frac{d}{du}%
\right) ^{n+2\nu }\left[ \Theta _{n+1,\nu }(t,u)\right] 
\end{split}
\end{equation}
where%
\begin{equation*}
\Theta _{n+1,\nu }(t,u):=\sum_{m=0}^{\infty }e^{-4t(m+\nu +\frac{n}{2}%
)^{2}}\cos (2m+2\nu +n)u.
\end{equation*}
\end{theorem}
\begin{proof}
We first write the heat kernel in terms of the reproducing kernel \eqref{RK} as follows 
\begin{equation}\label{HK=sum(expK)}
\begin{split}
\quad H_{\nu }(t,z,w)&=\sum\limits_{m=0}^{+\infty }e^{\beta_{m}t}K_{\nu,m}(z,w)\\
&=\sum\limits_{m=0}^{+\infty }e^{\beta_{m}t}\frac{(2m+2\nu +n)\Gamma (m+2\nu +n)}{\pi^{n}\Gamma (m+2\nu +1)}\left( \frac{\left( 1+\langle w,z\rangle \right) ^{2}}{(1+|z|^{2})(1+|w|^{2})}\right) ^{\nu } \\
 & \qquad\qquad\qquad\times P_{m}^{(n-1,2\nu )}\left( \frac{2|1+\langle z,w\rangle |^{2}}{(1+|z|^{2})(1+|w|^{2})}-1\right) .
\end{split}
\end{equation}%
Next, we will use the integral representation for Jacobi polynomials in terms of Gegenbauer polynomials given in \cite[p.194]{DK1971} followed by the symmetry relation $P_{k}^{(\alpha ,\beta
)}(-x)=(-1)^{k}P_{k}^{(\beta ,\alpha )}(x)$, 
\begin{equation}\label{IR_Jac-Geg}
P_{k}^{(\alpha ,\beta )}(2t^{2}-1)=\frac{2\Gamma (\alpha +\beta +1)\Gamma
(k+\beta +1)}{\Gamma (\beta +\frac{1}{2})\Gamma (k+\alpha +\beta +1)}%
\int_{0}^{1}(1-u^{2})^{\beta -\frac{1}{2}}C_{2k}^{\alpha +\beta +1}(tu)du,
\end{equation}%
$Re\alpha >-\frac{1}{2},\ Re\beta >-\frac{1}{2}$ and $|t|<1$. Precisely, we set $t=\cos \rho $ and $u=\cos \psi /\cos \rho $, then we specialize $\alpha =n-1,\ \beta =2\nu$ and $ k=m$ in \eqref{IR_Jac-Geg} to get that 
\begin{equation}
P_{m}^{(n-1 ,2\nu )}(\cos 2\rho)=\frac{2\Gamma (n+2\nu )\Gamma (m+2\nu
+1)}{\Gamma (2\nu +\frac{1}{2})\Gamma (m+2\nu +n)}\frac{1}{\cos ^{4\nu }\rho}
\int_{\rho}^{\frac{\pi }{2}}\frac{\sin \psi }{\left( \cos ^{2}\rho-\cos
^{2}\psi \right) ^{\frac{1}{2}-2\nu }}C_{2m}^{n+2\nu }(\cos \psi
)d\psi .
\end{equation}%
By substituting the above representation in \eqref{HK=sum(expK)}, we obtain, after some calculations, the following
\begin{equation}  \label{Q_int}
\begin{split}
H_{\nu }(t,z,w)& =\frac{2\Gamma (n+2\nu )e^{4t(\nu^2+\frac{n^2}{4})}%
}{\Gamma (2\nu +\frac{1}{2})\cos ^{4\nu }\rho}D_{\nu }(z,w) \\
& \times \int_{\rho}^{\frac{\pi }{2}}\frac{\sin \psi }{\left( \cos
^{2}\theta -\cos ^{2}\psi \right) ^{\frac{1}{2}-2\nu }}\left(
\sum\limits_{m=0}^{+\infty }(2m+2\nu +n)C_{2m}^{n+2\nu }(\cos \psi
)e^{-4t(m+\nu +\frac{n}{2})^2}\right) d\psi ,
\end{split}%
\end{equation}
where 
\begin{equation*}
D_{\nu }(z,w):=\left(\frac{(1+\langle w,z\rangle )^2\cos^{-4}\rho}{(1+|z|^{2})(1+|w|^{2})} \right)^{\nu }.
\end{equation*}
Next, we use the formula \cite[p.876]{HI2002}: 
\begin{equation}
\left( -\frac{d}{\sin udu}\right) ^{l}\left( \frac{\sin (k+l+1)u}{\sin u}%
\right) =2^{l}l!C_{k}^{l+1}(\cos u),
\end{equation}%
for $l=n+2\nu ,\ k=2m,\ u=\psi $ to rewrite the series occurring in \eqref{Q_int} as 
\begin{equation}  \label{der_sum}
\frac{2^{1-n-2\nu }}{(n+2\nu -1)!}\left( -\frac{d}{\sin \psi d\psi }%
\right) ^{n-1+2\nu }\left( \sum\limits_{m=0}^{+\infty }\frac{(2m+2\nu
+n)\sin (2m+2\nu +n)\psi }{\sin \psi }e^{-4t(m+\nu +\frac{n}{2})^2}\right) .
\end{equation}
But since 
\begin{equation}
\frac{(2m+2\nu +n)\sin (2m+2\nu +n)\psi }{\sin \psi }=-\frac{d}{\sin
\psi d\psi }(\cos (2m+2\nu +n)\psi ).
\end{equation}
Equation \eqref{der_sum} may be presented as 
\begin{equation}  \label{der_theta}
\frac{2^{1-n-2\nu }}{(n+2\nu -1)!}\left( -\frac{d}{\sin \psi d\psi }%
\right) ^{n+2\nu }\left[ \Theta _{n+1,\nu }(t,\psi )\right]
\end{equation}
where 
\begin{equation}
\Theta _{n+1,\nu }(t,\psi )=\sum\limits_{m=0}^{+\infty}e^{-4t(m+\nu +\frac{n}{2})^2}\cos (2m+2\nu+n)\psi.
\end{equation}
Finally, by substituting \eqref{der_theta} into \eqref{Q_int} we arrive at the announced result \eqref{IRHK-AMH}.
\end{proof}

\section{Heat coefficients for $\Delta _{\protect\nu }$}
Here our goal is to establish an asymptotic expansion for the trace of $\exp\left(\frac{1}{4}t\Delta _{\nu }\right)$. For that, we start from the following spectral series representation of this trace 
\begin{equation}\label{HTF}
Tr\left( \exp\left(\frac{1}{4}t\Delta _{\nu }\right)\right) =e^{\left( \frac{n^{2}}{4}%
+\nu ^{2}\right) t}\sum\limits_{m=0}^{+\infty }\left( \dim _{\mathbb{C}}%
\mathcal{A}_{m}^{\nu }\right) e^{-\left( m+\frac{n}{2}+\nu \right) ^{2}t}, 
\end{equation}
where the dimension of $\mathcal{A}_{m}^{\nu }$ can be computed via the diagonal of the corresponding reproducing kernel as 
\begin{equation}\label{dim_A_m_nu}
\begin{split}
\dim _{\mathbb{C}}\mathcal{A}_{m}^{\nu } :=\int_{\mathbb{C}^n}K_{\nu,m}(z,z)d\mu_{n}(z)=\frac{(2m+n+2\nu )\Gamma
(m+n)\Gamma (m+n+2\nu )}{n\left( \Gamma (n)\right) ^{2}\Gamma (m+1)\Gamma
(m+2\nu +1)}.
\end{split}
\end{equation}
\begin{theorem}
The heat trace formula \eqref{HTF} satisfies the asymptotic expansion
\begin{equation}
Tr\left( \exp\left(\frac{1}{4}t\Delta _{\nu }\right)\right) \simeq \frac{1}{(4\pi t)^{n}}%
\sum\limits_{j=0}^{+\infty }b_{j}^{\left( \nu ,n\right) }\ t^{j},\text{ \  \ 
}t\searrow 0^{+} 
\end{equation}
with the heat coefficients given by
\begin{equation*}
b_{j}^{\left( \nu ,n\right) }=\frac{(4\pi)^n}{n!}
\sum\limits_{i=0}^{j}\frac{\left( \frac{n^{2}}{4}+\nu ^{2}\right) ^{j-i}}{\left( j-i\right) !}c_{i}^{\left( \nu ,n\right) },\text{ \ }
\end{equation*}
\begin{itemize}
\item[(i)] case $n$ odd
\begin{equation}
c_{i}^{\left( \nu ,n\right) }=\left\{ 
\begin{array}{c}
\frac{(n-i-1)!}{(n-1)!}\gamma _{n-i-1}^{(\nu ,n)},\quad\qquad\qquad\qquad\qquad 0\leq i\leq n-1 \\ 
\sum\limits_{p=0}^{n-1}\frac{(-1)^{i-n+1}c_{p}^{\left( \nu ,n\right)}}{(i-n)!(i-p)p!} B_{2(i-p)}\left(\nu+\frac{1}{2}\right)
,\quad i\geq n%
\end{array}%
\right.   
\end{equation}

\item[(ii)] case $n$ even
\begin{eqnarray}
\begin{split}
c_{i}^{\left( \nu ,n\right) }=\left\{
\begin{array}{c}
1,\qquad\qquad\qquad\qquad\qquad\qquad\qquad i=0\\
\tau _{n-1-i}^{(\nu ,n)}\frac{\left(n-i-1\right) !}{\left( n-1\right) !},\qquad\qquad\qquad\qquad\qquad 1\leq i\leq n-1\\
\sum\limits_{p=0}^{n-1}\frac{(-1)^{i-n}c_{p}^{\left( \nu ,n\right)}}{(i-n)!(i-p)(n-p-1)!}\left[ 2B_{2(i-p)}-B_{2(i-p)}(\nu)\right] ,\qquad i\geq n
\end{array}
\right.
\end{split}
\end{eqnarray}
\end{itemize} 
where $B_d$ and $B_d(\cdot)$ denote Bernoulli numbers and polynomials respectively (see Appendix \ref{appendix1}). The coefficients $\tau _{i}^{(\nu ,n)}$ and $\gamma_{i}^{(\nu ,n)}$ are defined by \eqref{dim_sum_1} and \eqref{dim_A_case2} respectively.
\end{theorem}

\begin{proof} We may rewrite \eqref{dim_A_m_nu} as
\begin{equation}\label{dim_A_prod}
\dim _{\mathbb{C}}\mathcal{A}_{m}^{\nu }=\frac{2m+n+2\nu }{n!(n-1)!}%
\prod\limits_{j=1}^{n-1}(m+j)(m+2\nu +j) .
\end{equation}
Setting $r=r\left( m\right) =m+\nu +\frac{n}{2}$, the product in \eqref{dim_A_prod} also reads
\begin{equation}\label{wp}
\mathcal{\wp }=\prod\limits_{j=1}^{n-1}\left( r-\frac{n}{2}-\nu +j\right) \left( r-\frac{n}{%
2}+\nu +j\right) .
\end{equation}%
We treat the cases of $n$ odd and $n$ even separately as
they correspond to different choices of theta functions.\\

$\left( i\right) $\textbf{\ \textit{Case odd} $n\geq 0$.}  The product \eqref{wp} can also be presented as 
\begin{equation}\label{dim_prod_2}
\mathcal{\wp }=\prod\limits_{l=\frac{1}{2}+\nu }^{\frac{n}{2}+\nu -1}\left(
r^{2}-l^{2}\right) \prod\limits_{s=\frac{1}{2}-\nu }^{\frac{n}{2}-\nu
-1}\left( r^{2}-s^{2}\right) 
\end{equation}%
and further it can be decomposed as 
\begin{equation}\label{dim_sum_1}
\mathcal{\wp }=\sum\limits_{p=0}^{n-1}\gamma _{p}^{(\nu ,n)}r^{2p}. 
\end{equation}%
where the numbers $\gamma _{0}^{(\nu ,n)},...,\gamma _{n-1}^{(\nu ,n)}$ can
be computed explicitly (see Sect.\ref{Sect8}). By using \eqref{dim_sum_1},
the multiplicity in \eqref{dim_A_prod} takes the form

\begin{equation}\label{dim_A_sum_2}
\dim _{\mathbb{C}}\mathcal{A}_{m}^{\nu }=\frac{2}{n!(n-1)!}%
\sum\limits_{p=0}^{n-1}\gamma _{p}^{(\nu ,n)}\left( m+\nu +\frac{n}{2}%
\right) ^{2p+1}
\end{equation}%
which we may use to rewrite the trace formula \eqref{HTF} as
\begin{eqnarray}
\begin{split}
Tr\left( e^{\frac{1}{4}t\Delta _{\nu }}\right) &=e^{\left( \frac{n^{2}%
}{4}+\nu ^{2}\right) t}\sum\limits_{m=0}^{+\infty }\left[ \frac{2}{n!(n-1)!}%
\sum\limits_{p=0}^{n-1}\gamma _{p}^{(\nu ,n)}\left( m+\nu +\frac{n}{2}%
\right) ^{2p+1}\right] e^{-\left( m+\frac{n}{2}+\nu \right) ^{2}t} \\
&=\frac{e^{\left( \frac{n^{2}}{4}+\nu ^{2}\right) t}}{n!(n-1)!}\mathcal{S} 
\end{split}
\end{eqnarray}
with 
\begin{equation}
\mathcal{S}:=\sum\limits_{p=0}^{n-1}\gamma _{p}^{(\nu ,n)}\left(
2\sum\limits_{m=0}^{+\infty }\left( m+\nu +\frac{n}{2}\right)
^{2p+1}e^{-\left( m+\frac{n}{2}+\nu \right) ^{2}t}\right) .  
\end{equation}%
The latter sum can be cast into two parts as 
\begin{equation}\label{dim_double_sum}
\mathcal{S}=2\sum\limits_{\mu =\nu +\frac{1}{2}}^{+\infty }\mu e^{-\mu
^{2}t}\sum\limits_{p=0}^{n-1}\gamma _{p}^{(\nu ,n)}\mu
^{2p}-2\sum\limits_{\mu =\nu +\frac{1}{2}}^{\nu +\frac{n}{2}-1}\mu e^{-\mu
^{2}t}\sum\limits_{p=0}^{n-1}\gamma _{p}^{(\nu ,n)}\mu ^{2p}  
\end{equation}%
Recalling \eqref{dim_prod_2}-\eqref{dim_sum_1}, we see that for $\mu $ raining over the set $\{\nu +\frac{1}{2},\nu +\frac{3}{2},\cdots ,\nu +\frac{n}{2}-1\}$ the sum%
\begin{equation}
\sum\limits_{p=0}^{n-1}\gamma _{p}^{(\nu ,n)}\mu ^{2p}=0
\end{equation}
and therefore the sum \eqref{dim_double_sum} reduces to%
\begin{equation}\label{sum_S_R}
\mathcal{S}=2\sum\limits_{p=0}^{n-1}\gamma _{p}^{(\nu ,n)}\mathcal{R}%
_{p}^{\left( \nu \right) }\left( t\right) 
\end{equation}
where 
\begin{equation}
\mathcal{R}_{p}^{\left( \nu \right) }\left( t\right) =\sum\limits_{\mu =\nu +%
\frac{1}{2}}^{+\infty }\mu ^{2p+1}e^{-\mu ^{2}t} 
\end{equation}%
By setting $\mu =j+\frac{1}{2},$ we successively obtain
\begin{eqnarray}
\mathcal{R}_{p}^{\left( \nu \right) }\left( t\right) &=&\sum\limits_{j=\nu
}^{+\infty }\left( j+\frac{1}{2}\right) ^{2p+1}e^{-\left( j+\frac{1}{2}%
\right) ^{2}t} \nonumber \\
&=&\sum\limits_{j=0}^{+\infty }\left( j+\frac{1}{2}\right) ^{2p+1}e^{-\left( j+%
\frac{1}{2}\right) ^{2}t}-\sum\limits_{j=0}^{\nu -1}\left( j+\frac{1}{2}%
\right) ^{2p+1}e^{-\left( j+\frac{1}{2}\right) ^{2}t} \nonumber \\
&=&\frac{1}{2}\left( -\frac{d}{dt}\right) ^{p}\left[ \sum\limits_{j=0}^{+%
\infty }\left( 2j+1\right) e^{-\left( j+\frac{1}{2}\right) ^{2}t}\right]
-\sum\limits_{j=0}^{\nu -1}\left( j+\frac{1}{2}\right) ^{2p+1}e^{-\left( j+%
\frac{1}{2}\right) ^{2}t} \nonumber\\
&=&\frac{1}{2}\left( -\frac{d}{dt}\right) ^{p}\left[ \vartheta _{2}(t)\right]
-\sum\limits_{j=0}^{\nu -1}\left( j+\frac{1}{2}\right) ^{2p+1}e^{-\left( j+%
\frac{1}{2}\right) ^{2}t} \label{finite_sum}
\end{eqnarray}
where $\vartheta _{2}(t)$ denotes the Jacobi's theta function given by \cite[p.280]{M_1928}: 
\begin{equation}
\vartheta _{2}(t)=\sum\limits_{j=0}^{+\infty }\left( 2j+1\right) e^{-\left(
j+\frac{1}{2}\right) ^{2}t}  
\end{equation}%
while the finite sum in \eqref{finite_sum} may be expressed by the series%
\begin{equation}
\sum\limits_{\ell =0}^{+\infty }\sigma _{p}^{(\nu )}\left( \ell \right) \
t^{\ell }  
\end{equation}%
with coefficients 
\begin{equation}
\sigma _{p}^{\left( \nu \right) }\left( \ell \right) =\frac{(-1)\ell }{\ell !%
}\sum\limits_{j=0}^{\nu -1}\left( j+\frac{1}{2}\right) ^{2(p+\ell )+1}. 
\end{equation}%
Now, we use the formula \cite[p.597]{prud1}:
\begin{equation}
\sum\limits_{k=0}^{m}(k+a)^q=\frac{1}{q+1}\left[ B_{q+1}(m+1+a)-B_{q+1}(a)\right], \quad q=1,2,\cdots ,
\end{equation}
where $B_n(z)$ are the Bernoulli polynomials, for $m=\nu-1$, $q=2(\ell+p)+1$, $a=\frac{1}{2}$, to get 
\begin{equation}
\sigma _{p}^{\left( \nu \right) }\left( \ell \right) =\frac{(-1)^{\ell}}{2(p+\ell +1)\ell !}\left[ B_{2(p+\ell +1)}\left(\nu+\frac{1}{2}\right)-B_{2(p+\ell +1)}\left(\frac{1}{2}\right)\right]. 
\end{equation}
Summarizing the above calculations, then the sum in \eqref{sum_S_R} also reads
\begin{equation}\label{S_TJ}
\mathcal{S}=\sum\limits_{p=0}^{n-1}\gamma _{p}^{(\nu ,n)}\left[ \left( -%
\frac{d}{dt}\right) ^{p}\vartheta _{2}(t)+\sum\limits_{\ell =0}^{+\infty}\frac{(-1)^{\ell}}{(p+\ell +1)\ell !}\left[ B_{2(p+\ell +1)}\left(\frac{1}{2}\right)-B_{2(p+\ell +1)}\left(\nu+\frac{1}{2}\right)\right] \ t^{\ell }\right] 
\end{equation}
Applying now the asymptotic of the higher order derivative of $\vartheta _{2}(t)$
as $t\searrow 0^{+}$ (\cite[p.281]{M_1928}) :
\begin{equation}\label{TJ_p_deri}
\left( \frac{d}{dt}\right) ^{p}\vartheta _{2}(t)\simeq \frac{(-1)^{p}p!}{%
t^{1+p}}+\sum\limits_{s=0}^{+\infty }\frac{B_{s+p}t^{s}}{s!}  
\end{equation}%
in terms of Bernoulli numbers $\left( B_{d\text{\ }}\right) $ defined by the
recursion relation \cite[p.8]{Aw2019}:
\begin{equation}\label{Ber_num}
B_{d}=\frac{\left( -1\right) ^{d}}{d+1}\left( 1-2^{-2d-1}\right)B_{2d+2} ,%
\text{ \ }d=0,1,2,...  
\end{equation}%
By combining \eqref{Ber_num} and the formula \cite[p.777]{prud1}
\begin{equation}
B_d\left(\frac{1}{2}\right) =-(1-2^{1-d}) B_d,
\end{equation}
we deduce that
\begin{equation}\label{Ber(1/2)-formula}
B_{2(d+1)}\left(\frac{1}{2}\right) =(-1)^{d+1}(d+1)B_d.
\end{equation}
For $d=p+s$ \eqref{Ber(1/2)-formula} reads
\begin{equation}\label{Ber(1/2)}
B_{2(p+s+1)}\left(\frac{1}{2}\right) =(-1)^{p+s+1}(p+s+1)B_{p+s}.
\end{equation}
Substituting \eqref{TJ_p_deri} and \eqref{Ber(1/2)} into \eqref{S_TJ}, we get after simplifications
\begin{equation}
\mathcal{S} \simeq \sum\limits_{p=0}^{n-1}\gamma _{p}^{(\nu ,n)}\left[
p!t^{-1-p}+\sum\limits_{s=0}^{+\infty }\frac{(-1)^{s+1}B_{2(p+s+1)}\left(\nu+\frac{1}{2}\right)}{(p+s+1)s!}t^{s}\right] .
\end{equation}
Therefore, the heat trace asymptotic formula becomes 
\begin{equation}
\begin{split}
Tr\left( e^{\frac{1}{4}t\Delta _{\nu }}\right) &\simeq \frac{e^{\left( \frac{%
n^{2}}{4}+\nu ^{2}\right) t}}{n!(n-1)!}\sum\limits_{p=0}^{n-1}\gamma
_{p}^{(\nu ,n)}\left[ \frac{p!}{t^{1+p}}+\sum\limits_{s=0}^{+\infty }
\frac{(-1)^{s+1}B_{2(p+s+1)}\left(\nu+\frac{1}{2}\right)}{(p+s+1)s!}
 t^{s}\right] \\
&= \frac{e^{\left( \frac{n^{2}}{4}+\nu ^{2}\right) t}}{n!(n-1)!}\left[ \sum\limits_{p=0}^{n-1}\gamma
_{p}^{(\nu ,n)}\frac{p!}{t^{1+p}}+\sum\limits_{p=0}^{n-1}\gamma
_{p}^{(\nu ,n)}\sum\limits_{s=0}^{+\infty }
\frac{(-1)^{s+1}B_{2(p+s+1)}\left(\nu+\frac{1}{2}\right)}{(p+s+1)s!}
 t^{s}\right] \\
& =\frac{\omega _{n}e^{\left( \frac{n^{2}}{4}+\nu ^{2}\right) t}}{(4\pi t)^{n}}%
\left[ 1+\sum\limits_{p=0}^{n-2}\frac{\gamma _{p}^{(\nu ,n)}p!}{(n-1)!}%
t^{n-p-1}+\sum\limits_{p=0}^{n-1}\frac{\gamma _{p}^{(\nu ,n)}}{(n-1)!}\sum\limits_{s=p}^{+\infty }\frac{(-1)^{s-p+1}B_{2(s+1)}\left(\nu+\frac{1}{2}\right)}{(s+1)(s-p)!} t^{n+s-p}\right]   \label{asym_TF_3part}
\end{split}
\end{equation}
The first sum in \eqref{asym_TF_3part} can also be written as%
\begin{equation}
\sum\limits_{d=0}^{n-1}\frac{\gamma _{n-1-d}^{(\nu ,n)}}{(n-1)!}\left(
n-1-d\right) !t^{d}  
\end{equation}%
while the second one also reads%
\begin{equation}
\sum\limits_{d=n}^{+\infty }\left( \sum\limits_{p=0}^{n-1}\frac{(-1)^{d-n+1}\gamma_{p}^{(\nu ,n)} B_{2(d-n+p+1)}\left(\nu+\frac{1}{2}\right)}{(n-1)!(d-n+p+1)(d-n)!} \right) t^{d}  .
\end{equation}
By introducing 
\begin{equation}
\quad \Omega _{p}^{\left( \nu \right) }\left( k\right) :=\frac{(-1)^{k+1} B_{2(k+p+1)}\left(\nu+\frac{1}{2}\right)}{(k+p+1)k!} ,\qquad  k=0,1,2,...,
\end{equation}
we may rewrite \eqref{asym_TF_3part} as%
\begin{equation}
Tr\left( e^{\frac{1}{4}t\Delta _{\nu }}\right) \simeq \frac{\omega
_{n}e^{\left( \frac{n^{2}}{4}+\nu ^{2}\right) t}}{(4\pi t)^{n}}\left[
1+\sum\limits_{d=1}^{n-1}\frac{\gamma _{n-d-1}^{(\nu ,n)}(n-d-1)!}{(n-1)!}%
t^{d}+\sum\limits_{d=n}^{+\infty }\left( \sum\limits_{p=0}^{n-1}\frac{\gamma
_{p}^{(\nu ,n)}}{(n-1)!}\Omega _{p}^{\left( \nu \right) }\left( d-n\right)
\right) t^{d}\right]  
\end{equation}

\begin{equation}\label{trace_prod_1}
=\frac{\omega _{n}}{(4\pi t)^{n}}\left[ \sum\limits_{d=0}^{+\infty
}c_{d}^{\left( \nu ,n\right) }\ t^{d}\right] \left[ \sum\limits_{i=0}^{+%
\infty }\frac{1}{i!}\left( \frac{n^{2}}{4}+\nu ^{2}\right) ^{i}\right]  
\end{equation}%
where the coefficients $c_{d}^{\left( \nu ,n\right) }$are given by  
\begin{equation}\label{coef_c_d}
c_{d}^{\left( \nu ,n\right) }=\left\{ 
\begin{array}{c}
\frac{(n-d-1)!}{(n-1)!}\gamma _{n-d-1}^{(\nu ,n)},\quad 0\leq d\leq n-1 \\ 
\sum\limits_{p=0}^{n-1}\frac{\gamma_{p}^{(\nu ,n)}}{(n-1)!}\Omega _{p}^{\left( \nu \right) }\left( d-n\right),\quad d\geq n%
\end{array}
\right.  
\end{equation}%
Since $p\in \{0,\cdots, n-1\}$ in the second expression of \eqref{coef_c_d}, we may Write
\begin{equation}
\gamma _{p}^{(\nu ,n)}=\frac{(n-1)!}{p!}c_{n-p-1}^{\left( \nu ,n\right)}, \qquad 0\leq p \leq n-1 
\end{equation}
and then, the expression of $c_{d}^{\left( \nu ,n\right) }$ can be simplify as follows
\begin{equation}
c_{d}^{\left( \nu ,n\right) }=\left\{ 
\begin{array}{c}
\frac{(n-d-1)!}{(n-1)!}\gamma _{n-d-1}^{(\nu ,n)},\quad 0\leq d\leq n-1 \\ 
\sum\limits_{p=0}^{n-1}\frac{c_{n-p-1}^{\left( \nu ,n\right)}}{p!}\Omega _{p}^{\left( \nu \right) }\left( d-n\right)
,\quad d\geq n.
\end{array}%
\right.  
\end{equation}
This means that the coefficients $c_{d}^{\left( \nu ,n\right) }$ satisfy a $n$-terms recurrence relation when $d\geq n$.\\

$\left( i\right) $\textbf{\ \textit{Case even}} $n\geq 2$. In this case, the product \eqref{wp} can also be presented as
\begin{equation}\label{dim_prod_case2}
\mathcal{\wp }'=\prod\limits_{l=\nu }^{\frac{n}{2}+\nu -1}\left(
r^{2}-l^{2}\right) \prod\limits_{s=1-\nu }^{\frac{n}{2}-\nu
-1}\left( r^{2}-s^{2}\right) \qquad \text{(product omitted for $n = 2$)}
\end{equation}%
and further it can be decomposed as 
\begin{equation}\label{dim_A_case2}
\mathcal{\wp }'=\sum\limits_{p=0}^{n-1}\tau _{p}^{(\nu ,n)}r^{2p}.  
\end{equation}
where the numbers $\tau _{0}^{(\nu ,n)},...,\tau _{n-1}^{(\nu ,n)}$ can be
computed explicitly (see Sect.\ref{Sect8}). By using \eqref{dim_A_case2},
the multiplicity in \eqref{dim_A_prod} takes the form

\begin{equation}\label{dim_A_sum_case2}
\dim _{\mathbb{C}}\mathcal{A}_{m}^{\nu }=\frac{2}{n!(n-1)!}%
\sum\limits_{p=0}^{n-1}\tau _{p}^{(\nu ,n)}\left( m+\nu +\frac{n}{2}\right) ^{2p+1}
\end{equation}
We now may insert the r.h.s of \eqref{dim_A_case2} into the trace formula \eqref{HTF} to get 
\begin{equation}
\begin{split}
Tr\left( e^{\frac{1}{4}t\Delta _{\nu }}\right) &= e^{\left( \frac{n^{2}%
}{4}+\nu ^{2}\right) t}\sum\limits_{m=0}^{+\infty }\left[ \frac{2}{n!(n-1)!}%
\sum\limits_{p=0}^{n-1}\tau _{p}^{(\nu ,n)}\left( m+\nu +\frac{n}{2}\right)^{2p+1}\right] e^{-\left( m+\nu +\frac{n}{2}\right)^{2}t}\\
&=\frac{2 e^{\left( \frac{n^{2}}{4}+\nu ^{2}\right) t}}{n!(n-1)!}\sum\limits_{\mu=\nu+\frac{n}{2}}^{+\infty }e^{-\mu^{2}t}\sum\limits_{p=0}^{n-1}\tau _{p}^{(\nu ,n)}\mu^{2p+1} \\
&=\frac{ e^{\left( \frac{n^{2}}{4}+\nu ^{2}\right) t}}{n!(n-1)!}\left[ 2\sum\limits_{\mu=1}^{+\infty }e^{-\mu^{2}t}\sum\limits_{p=0}^{n-1}\tau _{p}^{(\nu ,n)}\mu^{2p+1}-2\sum\limits_{\mu=1}^{\nu+\frac{n}{2}-1}e^{-\mu^{2}t}\sum\limits_{p=0}^{n-1}\tau _{p}^{(\nu ,n)}\mu^{2p+1}\right] \\
&=\frac{ e^{\left( \frac{n^{2}}{4}+\nu ^{2}\right) t}}{n!(n-1)!}\left[S_1-S_2\right].
\end{split} 
\end{equation}%
Now, recall the Jacobi's theta function $\left( \left[ ..\right] \text{, p..}%
\right) :$ 
\begin{equation}
\vartheta _{3}\left( t\right) :=2\sum\limits_{l=1}^{+\infty }le^{-l^{2}t} 
\end{equation}
Then one can see that the $p$ derivative of $\vartheta _{3}\left(
t\right) $ satisfies 
\begin{equation}
\left( -1\right) ^{p}\left( \frac{d}{dt}\right) ^{p}\left[ \vartheta
_{3}(t)\right] =2\sum\limits_{l=1}^{+\infty }l^{2p+1}e^{-l^{2}t} 
\end{equation}%
Therefore,
\begin{equation}\label{R_t_case2}
S_1=\sum\limits_{p=0}^{n-1}\tau _{p}^{(\nu ,n)}\left( -1\right)
^{p}\left( \frac{d}{dt}\right) ^{p}\left[ \vartheta _{3}(t)\right].
\end{equation}
From \eqref{dim_prod_case2} and \eqref{dim_A_case2} we see that for $\mu $ raning over the set $\{\nu ,\nu +1,\cdots ,\nu +\frac{n}{2}-1\}$ the sum
\begin{equation}
\sum\limits_{p=0}^{n-1}\tau _{p}^{(\nu ,n)}\mu ^{2p}=0,
\end{equation}
and therefore $S_2$ reduces to%
\begin{equation}
S_2=2\sum\limits_{\mu=1}^{\nu-1}e^{-\mu^{2}t}\sum\limits_{p=0}^{n-1}\tau _{p}^{(\nu ,n)}\mu^{2p+1}.
\end{equation}%
The above sum can be expressed as follows
\begin{equation}
\begin{split}
S_2&=2\sum\limits_{p=0}^{n-1}\tau _{p}^{(\nu ,n)}\sum\limits_{\mu=1}^{\nu-1}e^{-\mu^{2}t}\mu^{2p+1}\\
&=2\sum\limits_{p=0}^{n-1}\tau _{p}^{(\nu ,n)}\sum\limits_{\ell =0}^{+\infty }\varrho _{p}^{(\nu )}\left( \ell \right) \
t^{\ell }
\end{split}
\end{equation}%
with coefficients 
\begin{equation}
\varrho_{p}^{\left( \nu \right) }\left( \ell \right) =\frac{(-1)\ell }{\ell !%
}\sum\limits_{\mu=1}^{\nu -1}\mu ^{2(p+\ell )+1}. 
\end{equation}%
Now, we use the formula \cite[p.596]{prud1}:
\begin{equation}
\sum\limits_{k=1}^{m}k^q=\frac{1}{q+1}\left[ B_{q+1}(m+1)-B_{q+1}\right], \quad q=1,2,\cdots ,
\end{equation}
for $m=\nu-1$, $q=2(\ell+p)+1$, to get 
\begin{equation}\label{rho(l)}
\varrho_{p}^{\left( \nu \right) }\left( \ell \right) =\frac{(-1)^{\ell}}{2(p+\ell +1)\ell !}\left[ B_{2(p+\ell +1)}\left(\nu\right)-B_{2(p+\ell +1)}\right]. 
\end{equation}
Therefore, the trace formula becomes
\begin{equation}\label{HTF_2sum_case2}
\mathcal{T}:=Tr\left( e^{\frac{1}{4}t\Delta _{\nu }}\right) =\frac{e^{\left( \frac{n^{2}}{%
4}+\nu ^{2}\right) t}}{n!(n-1)!}\sum\limits_{p=0}^{n-1}\tau _{p}^{(\nu ,n)}%
\left[ \left( -1\right) ^{p}\left( \frac{d}{dt}\right) ^{p}\left[
\vartheta _{3}(t)\right] -\sum\limits_{\ell =0}^{+\infty }\frac{(-1)^{\ell}\left[ B_{2(p+\ell +1)}\left(\nu\right)-B_{2(p+\ell +1)}\right]}{(p+\ell +1)\ell !} \
t^{\ell } \right]  
\end{equation}
Applying now the asymptotic of the higher order derivative of $\vartheta _{3}(t)$
as $t\searrow 0^{+}$ \cite[p.11]{Awon2019} :
\begin{equation}\label{D_V3}
\left( \frac{d}{dt}\right) ^{\ell }\vartheta _{3}(t)\simeq \frac{(-1)^{\ell
}\ell !}{t^{1+\ell }}+\sum\limits_{j=\ell }^{+\infty }\frac{\left( -1\right) ^{j}B_{2j+2}}{(j+1)(j-\ell)!}t^{j-\ell }  
\end{equation}%
where $\left( B_{d}\right) $ are Bernoulli numbers in \eqref{Ber_num}, substituting \eqref{D_V3} into \eqref{HTF_2sum_case2}, we get

\begin{eqnarray*}
\begin{split}
\mathcal{T} &\simeq \frac{e^{\left( \frac{n^{2}}{4}+\nu ^{2}\right) t}}{n!(n-1)!}\sum\limits_{p=0}^{n-1}\tau
_{p}^{(\nu ,n)}\left[\frac{p !}{t^{1+p}}+\left( -1\right)^{p}\sum\limits_{j=p}^{+\infty }\frac{\left( -1\right) ^{j}B_{2j+2}}{(j+1)(j-p)!}t^{j-p }   -\sum\limits_{\ell =0}^{+\infty }\frac{(-1)^{\ell}\left[ B_{2(p+\ell +1)}\left(\nu\right)-B_{2(p+\ell +1)}\right]}{(p+\ell +1)\ell !} \ t^{\ell } \right]\\
&= \frac{e^{\left( \frac{n^{2}}{4}+\nu ^{2}\right) t}}{n!(n-1)!}\sum\limits_{p=0}^{n-1}\tau
_{p}^{(\nu ,n)}\left[\frac{p !}{t^{1+p}}+\sum\limits_{\ell=0}^{+\infty }\frac{\left( -1\right) ^{\ell}B_{2(\ell +p+1)}}{(\ell+p+1)\ell !}t^{\ell }   -\sum\limits_{\ell =0}^{+\infty }\frac{(-1)^{\ell}\left[ B_{2(p+\ell +1)}\left(\nu\right)-B_{2(p+\ell +1)}\right]}{(p+\ell +1)\ell !} \ t^{\ell } \right]  \\
&= \frac{e^{\left( \frac{n^{2}}{4}+\nu ^{2}\right) t}}{n!(n-1)!}\sum\limits_{p=0}^{n-1}\tau
_{p}^{(\nu ,n)}\left[\frac{p !}{t^{1+p}}+\sum\limits_{\ell =0}^{+\infty }\frac{(-1)^{\ell}\left[ 2B_{2(p+\ell +1)}-B_{2(p+\ell +1)}\left(\nu\right)\right]}{(p+\ell +1)\ell !} \ t^{\ell } \right]  
\end{split}
\end{eqnarray*}
We denote $Vol\left(\textbf{P}^n(\mathbb{C})\right)=\omega_n=(4\pi)^n/n!$ and we use the fact that $\tau_{n-1}^{(\nu ,n)}=1$, to write
\begin{eqnarray*}
\begin{split}
\mathcal{T} &\simeq \frac{\omega_n e^{\left( \frac{n^{2}}{4}+\nu ^{2}\right) t}}{(4\pi t)^n}\sum\limits_{p=0}^{n-1}\frac{\tau_{p}^{(\nu ,n)}}{(n-1)!}\left[p! t^{n-p-1}+\sum\limits_{\ell=0}^{+\infty }\frac{(-1)^\ell}{(p+\ell+1)\ell !}\left[ 2B_{2(p+\ell+1)}-B_{2(p+\ell+1)}(\nu)\right]  \ t^{n+\ell } \right] \\
&= \frac{\omega_n e^{\left( \frac{n^{2}}{4}+\nu ^{2}\right) t}}{(4\pi t)^n}\left[1+\sum\limits_{p=0}^{n-2}\frac{\tau_{p}^{(\nu ,n)}p !}{(n-1)!} t^{n-p-1}+\sum\limits_{p=0}^{n-1}\frac{\tau_{p}^{(\nu ,n)}}{(n-1)!}\sum\limits_{\ell=0}^{+\infty }\frac{(-1)^\ell}{(p+\ell+1)\ell !}\left[ 2B_{2(p+\ell+1)}-B_{2(p+\ell+1)}(\nu)\right]  \ t^{n+\ell } \right] \\
&= \frac{\omega_n e^{\left( \frac{n^{2}}{4}+\nu ^{2}\right) t}}{(4\pi t)^n}\left[1+T_1+T_2 \right] 
\end{split}
\end{eqnarray*}
For $T_1$, we have
\begin{equation}
T_1=\sum\limits_{p=0}^{n-2}\frac{\tau_{p}^{(\nu ,n)}p !}{(n-1)!} t^{n-p-1}=\sum\limits_{d=1}^{n-1}\frac{\tau_{n-d-1}^{(\nu ,n)}\left(n-d-1\right) !}{(n-1)!} t^{d}.
\end{equation}
For $T_2$, we have
\begin{eqnarray}
\begin{split}
S_2 &=\sum\limits_{p=0}^{n-1}\frac{\tau_{p}^{(\nu ,n)}}{(n-1)!}\sum\limits_{\ell=0}^{+\infty }\frac{(-1)^\ell}{(p+\ell+1)\ell !}\left[ 2B_{2(p+\ell+1)}-B_{2(p+\ell+1)}(\nu)\right]  \ t^{n+\ell }\\
&=\sum\limits_{p=0}^{n-1}\frac{\tau_{p}^{(\nu ,n)}}{(n-1)!}\sum\limits_{d=n}^{+\infty }\frac{(-1)^{d-n}}{(p+d-n+1)(d-n)!}\left[ 2B_{2(p+d-n+1)}-B_{2(p+d-n+1)}(\nu)\right]  \ t^{d }\\
&=\sum\limits_{d=n}^{+\infty }\left(\sum\limits_{p=0}^{n-1}\frac{(-1)^{d-n}\tau_{p}^{(\nu ,n)}\left[ 2B_{2(p+d-n+1)}-B_{2(p+d-n+1)}(\nu)\right]}{(n-1)!(p+d-n+1)(d-n)!} \right) \ t^{d }
\end{split}
\end{eqnarray}
Therefore
\begin{eqnarray*}
\begin{split}
\mathcal{T} &\simeq  \frac{\omega_n e^{\left( \frac{n^{2}}{4}+\nu ^{2}\right) t}}{(4\pi t)^n}\left[1+\sum\limits_{d=1}^{n-1}\frac{\tau_{n-d-1}^{(\nu ,n)}\left(n-d-1\right) !}{(n-1)!} t^{d}+\sum\limits_{d=n}^{+\infty }\left(\sum\limits_{p=0}^{n-1}\frac{(-1)^{d-n}\tau_{p}^{(\nu ,n)}\left[ 2B_{2(p+d-n+1)}-B_{2(p+d-n+1)}(\nu)\right]}{(n-1)!(p+d-n+1)(d-n)!} \right) \ t^{d } \right] 
\end{split}
\end{eqnarray*}
which can also be written as
\begin{equation}\label{trace_exp-sum}
\frac{\omega _{n}}{(4\pi t)^{n}}e^{\left( \frac{n^{2}}{4}+\nu ^{2}\right)
t}\sum\limits_{i=0}^{+\infty }c_{i}^{\left( \nu ,n\right) }t^{i}
\end{equation}%
where
\begin{eqnarray}
\begin{split}
c_{i}^{\left( \nu ,n\right) }=\left\{
\begin{array}{c}
1,\qquad\qquad\qquad\qquad\qquad\qquad\qquad i=0\\
\tau _{n-1-i}^{(\nu ,n)}\frac{\left(n-i-1\right) !}{\left( n-1\right) !},\qquad\qquad\qquad\qquad\qquad 1\leq i\leq n-1\\
\sum\limits_{p=0}^{n-1}\frac{(-1)^{i-n}\tau_{p}^{(\nu ,n)}\left[ 2B_{2(p+i-n+1)}-B_{2(p+i-n+1)}(\nu)\right]}{(n-1)!(p+i-n+1)(i-n)!} ,\qquad i\geq n
\end{array}
\right.
\end{split}
\end{eqnarray}
Since in the last expression of $c_{i}^{\left( \nu ,n\right) }$ where $i\geq n$, we have $p\in \{0,\cdots, n-1\}$ we may write
\begin{equation}
\gamma _{p}^{(\nu ,n)}=\frac{(n-1)!}{p!}c_{n-p-1}^{\left( \nu ,n\right)}, \qquad 0\leq p \leq n-1 
\end{equation}
and then, the expression of $c_{i}^{\left( \nu ,n\right) }$ simplifies as follows
\begin{eqnarray}
\begin{split}
c_{i}^{\left( \nu ,n\right) }=\left\{
\begin{array}{c}
1,\qquad\qquad\qquad\qquad\qquad\qquad\qquad i=0\\
\tau _{n-1-i}^{(\nu ,n)}\frac{\left(n-i-1\right) !}{\left( n-1\right) !},\qquad\qquad\qquad\qquad\qquad 1\leq i\leq n-1\\
\sum\limits_{p=0}^{n-1}\frac{(-1)^{i-n}c_{n-p-1}^{\left( \nu ,n\right)}\left[ 2B_{2(p+i-n+1)}-B_{2(p+i-n+1)}(\nu)\right]}{p!(p+i-n+1)(i-n)!} ,\qquad i\geq n.
\end{array}
\right.
\end{split}
\end{eqnarray}
Therefore, we may rewrite \eqref{trace_exp-sum} as%
\begin{equation}\label{trace_prod_2}
Tr\left( e^{\frac{1}{4}t\Delta _{\nu }}\right) \simeq \frac{\omega _{n}}{%
(4\pi t)^{n}}\left[ \sum\limits_{k=0}^{+\infty }c_{k}^{\left( \nu ,n\right)
}\ t^{k}\right] \left[ \sum\limits_{j=0}^{+\infty }\frac{1}{j!}\left( \frac{%
n^{2}}{4}+\nu ^{2}\right) ^{j}\right] .
\end{equation}%
Finally, in both cases the r.h.s of \eqref{trace_prod_1} and \eqref{trace_prod_2} may also be seen as a Cauchy product of two power series as   
\begin{equation*}
Tr\left( e^{\frac{1}{4}t\Delta _{\nu }}\right) \simeq\frac{1}{(4\pi t)^{n}}%
\sum\limits_{d=0}^{+\infty }\left[ \omega _{n}\sum\limits_{i=0}^{d}\frac{%
\left( n^{2}+4\nu ^{2}\right) ^{d-i}}{4^{d-i}\left( d-i\right) !}%
c_{i}^{\left( \nu ,n\right) }\right] t^{d}.
\end{equation*}
This ends the proof of the theorem.
\end{proof}

\section{The coefficient $b_j^{(\nu , n)}$ for $n=1,2,3,4$}\label{Sect8}
In this section, we exhibit the heat coefficients $b_{j}^{\left( \nu ,n\right) }$ for the special cases $n=1,2,3,4$. For that, we first compute the coefficients $\gamma_i^{(\nu,n)}$ and $\tau_i^{(\nu,n)}$.
\begin{tabbing}
$\gamma_0^{(\nu,1)}=1$ \\
$\tau_0^{(\nu,2)}=-\nu^2$, $\tau_1^{(\nu,2)}=1$\\
$\gamma_0^{(\nu,3)}=-2\left(\frac{1}{4}+\nu^2\right)$, $\gamma_1^{(\nu,3)}$, $\gamma_2^{(\nu,3)}=1$\\
$\tau_0^{(\nu,4)}=\nu (\nu^2-1)$, $\tau_1^{(\nu,4)}=-\nu^2+2\nu+1$, $\tau_2^{(\nu,4)}=-\nu-2$, $\tau_3^{(\nu,4)}=1$.
\end{tabbing}

\begin{itemize}
\item For $n=1$, we have that
\begin{equation}\label{b_j^1}
b_{j}^{\left( \nu ,1\right) }=4\pi\left[\frac{\left( \frac{1}{4}+\nu ^{2}\right) ^{j}}{j !}+\sum\limits_{i=1}^{j}\frac{(-1)^i\left( \frac{1}{4}+\nu ^{2}\right) ^{j-i}}{\left( j-i\right) !i!} B_{2i}(\nu+1/2)\right]
\end{equation}
If $\nu=0$, then \eqref{b_j^1} reduces to 
\begin{equation*}
b_{j}^{\left( 0 ,1\right) }=4\pi\sum\limits_{i=0}^{j}\frac{\left( \frac{1}{4}\right) ^{j-i}u_i^1}{\left( j-i\right) !}
\end{equation*}
with
\begin{equation}
u_0^1=1, \quad u_i^1=\frac{B_{i-1}}{(i-1)!}, \ i\geq 1
\end{equation}
in agreement with the results obtained in \cite[p.8]{Awon2019}.
\item For $n=2$, we have that
\begin{equation}\label{b_j^1}
b_{j}^{\left( \nu ,2\right) }=8\pi^2\left[\frac{\left( 1+\nu ^{2}\right) ^{j}}{j!}-\frac{\nu^2\left( 1+\nu ^{2}\right) ^{j-1}}{\left( j-1\right) !}+\sum\limits_{i=2}^{j}\frac{\left( 1+\nu ^{2}\right) ^{j-i}}{\left( j-i\right) !} c_{i}^{(\nu,2)}\right]
\end{equation}
If $\nu=0$, then \eqref{b_j^1} reduces to 
\begin{equation*}
b_{j}^{\left( 0 ,2\right) }=8\pi^2\sum\limits_{i=0}^{j}\frac{\left( \frac{1}{4}\right) ^{j-i}u_i^2}{\left( j-i\right) !}
\end{equation*}
with
\begin{equation}
u_0^2=1, \quad u_1^2=0, \quad u_i^2=\frac{(-1)^{i}B_{2i}}{i(i-2)!}, \ i\geq 2
\end{equation}
in agreement with the results obtained in \cite[p.12]{Awon2019}.
\item For $n=3$, the heat coefficients read 
\begin{equation}\label{b_j^3}
b_{j}^{\left( \nu ,3\right) }=\frac{(4\pi)^3}{3!}
\sum\limits_{i=0}^{j}\frac{\left( \frac{9}{4}+\nu ^{2}\right) ^{j-i}}{\left( j-i\right) !}c_{i}^{\left( \nu ,3\right) },\text{ \ }
\end{equation}
where
\begin{eqnarray*}
\begin{split}
& c_{0}^{\left( \nu ,3\right) }=1, \quad c_{1}^{\left( \nu ,3\right) }=-\left(\frac{1}{4}+\nu^2\right),\quad c_{2}^{\left( \nu ,3\right) }=\frac{1}{2}\left(\frac{1}{4}-\nu^2\right)^2,\\
& c_{i}^{\left( \nu ,3\right) }=\sum\limits_{p=0}^{2}\frac{\left( -1\right)^{i-2}c_{2-p}^{\left( \nu ,3\right)}B_{2(p+i-2)}(\nu+1/2)}{(i+p-2)p!(i-3)!}
,\quad i\geq 3
\end{split}
\end{eqnarray*}
If $\nu=0$, then \eqref{b_j^3} becomes
\begin{equation*}
b_{j}^{\left( 0 ,3\right) }=\frac{(4\pi)^3}{3!}
\sum\limits_{i=0}^{j}\frac{\left( \frac{9}{4}\right) ^{j-i}}{\left( j-i\right) !}u_{i}^3 ,\text{ \ }
\end{equation*}
with
\begin{eqnarray*}
\begin{split}
& u_{0}^3=1, \quad u_{1}^3=-\frac{1}{4},\quad u_{2}^3=\frac{1}{32},\\
& u_{i}^3=\sum\limits_{p=0}^{2}\frac{\left( -1\right)^{p}u_{2-p}^3 B_{i+p-3}}{p!(i-3)!}=\frac{1}{2(i-3)!}\left[B_{i-1}+\frac{1}{2}B_{i-2}+\frac{1}{16}B_{i-3}\right],\quad i\geq 3
\end{split}
\end{eqnarray*}
as expected for $n=3$ (see \cite[p.9]{Awon2019}).
\item For $n=4$, the heat coefficients read 
\begin{equation}\label{b_j^4}
b_{j}^{\left( \nu ,4\right) }=\frac{(4\pi)^4}{4!}
\sum\limits_{i=0}^{j}\frac{\left( 1+\nu ^{2}\right) ^{j-i}}{\left( j-i\right) !}c_{i}^{\left( \nu ,4\right) },\text{ \ }
\end{equation}
where
\begin{eqnarray*}
\begin{split}
& c_{0}^{\left( \nu ,4\right) }=1, \quad c_{1}^{\left( \nu ,4\right) }=-\frac{1}{3}\left(\nu+2\right),\quad c_{2}^{\left( \nu ,4\right) }=\frac{1}{6}\left(-\nu^2+2\nu+1\right),\quad c_{3}^{\left( \nu ,4\right) }=\frac{\nu}{6}(\nu^2-1),\\
& c_{i}^{\left( \nu ,4\right) }=\sum\limits_{p=0}^{3}\frac{\left( -1\right)^{i-4}c_{3-p}^{\left( \nu ,4\right)}}{(i+p-3)p!(i-4)!}\left[2 B_{2(p+i-3)}-B_{2(p+i-3)}(\nu) \right],\quad i\geq 4
\end{split}
\end{eqnarray*}
If $\nu=0$, then \eqref{b_j^4} becomes
\begin{equation*}
b_{j}^{\left( 0 ,4\right) }=\frac{(4\pi)^4}{4!}
\sum\limits_{i=0}^{j}\frac{u_{i}^4}{\left( j-i\right) !} ,\text{ \ }
\end{equation*}
with
\begin{eqnarray*}
\begin{split}
& u_{0}^4=1, \quad u_{1}^4=-\frac{2}{3},\quad u_{2}^4=\frac{1}{6},\quad u_{3}^4=0,\\
& u_{i}^3=\sum\limits_{p=0}^{3}\frac{\left( -1\right)^{i-4}u_{3-p}^4}{(i+p-3)p!(i-4)!}B_{2(p+i-3)},\quad i\geq 4
\end{split}
\end{eqnarray*}
as expected for $n=4$ (see \cite[p.13]{Awon2019}).
\end{itemize}

\begin{appendix}
\section{Appendix A}\label{appendix1}
Here we list the basic notations and definitions
of some special functions and orthogonal polynomials we have used in this
paper. For more details on the theory of these functions we refer to 
\cite{Chihara,Szego,Erdelyi,Srivastava3,Appel}.
\begin{enumerate}
\item For $a\in \mathbb{C}$, the shifted factorial or Pochhammer
symbol is defined by 
\begin{equation}
(a)_{k}=a(a+1)\dots (a+k-1),\qquad k\in \mathbb{N},  \label{pochhamer}
\end{equation}%
where by convention $(a)_{0}=1$. When $a=-n$ with $n\in \mathbb{N}^{\ast }=%
\mathbb{N}-\{0\}$, 
\begin{equation}
(-n)_{k}=\left\{ 
\begin{array}{cl}
\frac{(-1)^{k}n!}{(n-k)!},\qquad  & 0\leq k\leq n, \\ 
0,\qquad  & k>n.%
\end{array}%
\right.   \label{Pochhammer-n}
\end{equation}
\item For $a\in \mathbb{C}$ and $k\in \mathbb{N}$, the binomial
coefficient is defined by%
\begin{equation*}
\begin{pmatrix}
\alpha  \\ 
s%
\end{pmatrix}%
:=\alpha (\alpha -1)\cdots (\alpha -s+1)/s!\text{ \ if \ }s\in \mathbb{Z}%
_{+}\setminus \{0\}\text{ and }%
\begin{pmatrix}
\alpha  \\ 
0%
\end{pmatrix}%
:=1,\text{ for all }\alpha \in \mathbb{R}
\end{equation*}
\begin{equation}\label{binom}
\binom{a}{k}=\frac{a(a-1)\dots (a-k+1)}{k!}=\frac{(-1)^{k}(-a)_{k}}{k!}.
\end{equation}

\item For $\xi\in \mathbb{C}$, the gamma function is defined by 
\begin{equation}
\Gamma (\xi)=\int\limits_{0}^{\infty }t^{\xi-1}e^{-t}dt,\qquad \mathrm{Re}\,\xi>0.
\label{gamma}
\end{equation}%
Note that $\Gamma (n+1)=n!$ if $n\in \mathbb{N}$. The shifted factorial or Pochhammer symbol is defined by
\begin{equation}
(a)_{k}= a(a+1)\cdots (a+k-1),  \label{PochhammerGG}
\end{equation}%
and $(a)_k=\frac{\Gamma (a+k)}{\Gamma (a)}$ if $a\in \mathbb{C}\setminus \mathbb{Z}_{-}$.

\item For $a,b\in \mathbb{C}$ such that $\mbox{Re}\,a,\mbox{Re}\,b>0$,
the beta function is defined by 
\begin{equation}
\mathcal{B}(a,b)=\int\limits_{0}^{1}t^{a-1}(1-t)^{b-1}dt=\frac{\Gamma
(a)\Gamma (b)}{\Gamma (a+b)}.  \label{Beta}
\end{equation}

\item For $a_{1},\dots ,a_{p}\in \mathbb{C}$ and $c_{1},\dots
,c_{q}\in \mathbb{C}\setminus \mathbb{Z}_{-}$, the generalized
hypergeometric function is defined by the series 
\begin{equation}
{}_p F_q\left(\begin{matrix}a_{1},\cdots ,a_{p}\\c_{1},\cdots,c_{q}\end{matrix}\bigg| \xi\right)=\sum_{k=0}^{+\infty }\frac{(a_{1})_{k}\cdots (a_{p})_{k}}{%
(c_{1})_{k}\cdots (c_{q})_{k}}\frac{\xi^{k}}{k!},  \label{hypergeometric}
\end{equation}%
which terminates whenever at least one of the $p$ parameters $a_{i}$ equals $%
-1,-2,-3,\dots $. It converges for $|\xi|<\infty $ if $p\leq q$ or for $|\xi|<1$
if $p=q+1$ and it diverges for all $\xi\neq 0$ if $p>q+1$. Special cases of
this function are the Gauss hypergeometric function ${}_2 F_1\left(\begin{matrix}a, b\\ c\end{matrix}\bigg| \xi\right)$, the confluent hypergeometric function ${}_1 F_1\left(\begin{matrix}a\\ c\end{matrix}\bigg| \xi\right)$ and the binomial series ${}_1 F_0\left(\begin{matrix}a\\ -\end{matrix}\bigg| \xi\right)$. The later one reduces to $(1-\xi)^{-a}$ if $|\xi|<1$.\newline
is called the modified Bessel function of the first kind and order $a$. 
\item The Pfaff transformation Pfaff-Kummer
transformation \cite[p.33(19)]{Sr-Ma1984}: 
\begin{equation}\label{Pfaff_trans}
{}_{2}F{}_{1}\left( 
\begin{array}{c}
a,\,b \\ 
c%
\end{array}%
\left\vert \xi \right. \right) =(1-\xi )^{-a}{}_{2}F{}_{1}\left( 
\begin{array}{c}
a,\,c-b \\ 
c%
\end{array}%
\left\vert \frac{\xi }{1-\xi }\right. \right) ,\quad |\arg (1-\xi)|<\pi ,
\end{equation}
\item The Jacobi polynomial of parameters $a$ and $b$ is defined by 
\begin{equation}
P_{k}^{(a,b)}(x)=2^{-k}\sum_{j=0}^{k}\binom{k+a}{j}\binom{k+b}{k-j}%
(x+1)^{j}(x-1)^{k-j},\qquad a,b>-1.  \label{jacobi}
\end{equation}
In particular
\begin{equation}
P_{k}^{(a,b)}(1)=\frac{(a+1)_k}{k!}.
\end{equation}
The Jacobi polynomials can be expressed in terms of terminating ${}_{2}F_{1}$-series as 
\begin{equation}\label{Jacobi_Poly}
P_{k}^{(a,b)}(x)=\frac{\Gamma (b+k+1)}{k!\Gamma (b+1)}\left( \frac{x-1}{2}\right) ^{k}{}_2 F_1\left(\begin{matrix}-k,-a-k\\ b+1\end{matrix}\bigg| \frac{x+1}{x-1}\right), 
\end{equation}%
or 
\begin{equation}\label{2F1_Jacobi}
{}_{2}F{}_{1}\left( 
\begin{array}{c}
-k,\,b \\ 
c%
\end{array}%
\left\vert \xi \right. \right) =\frac{k!}{(c)_{k}}P_{k}^{(c-1,b-c-k)}(1-2\xi).
\end{equation}
\item The normalized Jacobi polynomials are defined by
\begin{equation*}
R_{k}^{(\alpha ,\beta )}(u):=\frac{P_{k}^{(\alpha ,\beta )}(u)}{%
P_{k}^{\alpha ,\beta }(1)}.
\end{equation*}
By \eqref{Pfaff_trans} and \eqref{2F1_Jacobi}, we get 
\begin{equation}\label{2F1_norm_Jacobi}
R_{k}^{(\alpha ,\beta )}(u)=\left( \frac{1+u}{2}\right)
^{k}{}_{2}F{}_{1}\left( 
\begin{array}{c}
-k,-k-\beta  \\ 
\alpha +1%
\end{array}%
\big|\frac{1-u}{1+u}\right) .
\end{equation}
\item The disk polynomials that were first studied by Zernike and
Brinkman \cite{Zernike}, \cite{Koor_1976} are given by 
\begin{equation}\label{Zernike_poly}
R_{p,q}^{\gamma }(\xi ):=|\xi |^{|p-q|}e^{i(p-q)\arg \xi }R_{\min
(p,q)}^{(\gamma ,|p-q|)}(2|\xi |^{2}-1).
\end{equation}
From \eqref{Zernike_poly} and \eqref{2F1_norm_Jacobi} by using the relation $%
\max (s,t)=|s-t|+\min (s,t)$, we can check easily that
\begin{equation}\label{2F1-R}
{}_{2}F{}_{1}\left( 
\begin{array}{c}
-s,-t \\ 
\gamma +1%
\end{array}%
\big|y\right) =(1-y)^{\frac{s+t}{2}}R_{s,t}^{\gamma }\left( (1-y)^{-\frac{1}{%
2}}\right) .
\end{equation}
\item Bernoulli number $B_d$ is defined as the coefficient of $x^d/d!$ in the series expansion of $x(e^x-1)^{-1}$, and the Bernoulli polynomials $B_d(z)$ are defined by $B_d(z)=\sum_{k=0}^{d}\binom{d}{k}B_kz^{d-k}$.
\end{enumerate}

\end{appendix}

%
%
%

\bigskip
${}^{*}$ Department of Mathematics, Faculty of
Sciences, Ibn Zohr University,\\ P.O. Box. 8106, Agadir, Morocco (\emph{ahbli.khalid@gmail.com}).\\
\\
$\sharp$ Centre R\'egional des M\'etiers de l'Education et de la Formation de K\'enitra, Morocco\\
(\emph{ hafoudaliali@gmail.com})\\
\\
${}^{\flat}$ Department of Mathematics, Faculty of Sciences and Technics (M'Ghila),\\ P.O. Box. 523, B\'{e}ni Mellal, Morocco (\emph{mouayn@gmail.com}).

\end{document}